\documentclass[journal]{IEEEtran}
\usepackage[pdftex]{graphicx}
\graphicspath{{../pdf/}{../jpeg/}}
\DeclareGraphicsExtensions{.pdf,.jpeg,.png}
\usepackage[cmex10]{amsmath}
\usepackage{url}
\usepackage[table,xcdraw]{xcolor}
\usepackage{graphicx,color,xspace}
\usepackage{graphicx}
\usepackage{amssymb,color,amsfonts,mathrsfs,balance,mathtools,xfrac}
\allowdisplaybreaks
\usepackage{subcaption}
\usepackage{cite}
\usepackage{siunitx}
\usepackage{blindtext}
\usepackage{tcolorbox}
\usepackage{xcolor}
\usepackage{lipsum}
\usepackage[font={small}]{caption}
\usepackage{algorithm}
\usepackage{algorithmic}
\usepackage{array} 
\usepackage{makecell}
\usepackage{tabularx}

\newtheorem{proof}{Proof}
\newtheorem{remark}{Remark}
\newtheorem{proposition}{Proposition}

\newcommand{\vc}[3]{\overset{#2}{\underset{#3}{#1}}}

\newcommand{\abc}{\text{AmBC}\xspace} 
\usepackage{soul}
\sethlcolor{yellow}

\begin{document}
\title{Enhancing AmBC Systems with Deep Learning for Joint Channel Estimation and Signal Detection}

\author{S. Zargari, A. Hakimi, C. Tellambura, \textit{Fellow, IEEE}, and A. Maaref, \IEEEmembership{Member, IEEE}
\thanks{This work was supported in part by Huawei Technologies Canada Company, Ltd.}
\thanks{Shayan Zargari, Azar Hakimi, and Chintha Tellambura are with the Department of Electrical and Computer Engineering, University of Alberta, Edmonton, AB T6G 1H9, Canada (e-mail: zargari@ualberta.ca; hakimina@ualberta.ca; ct4@ualberta.ca).}

\thanks{Amine Maaref is with the Ottawa Wireless Advanced System Competency Centre, Huawei Canada, Ottawa, ON K2K 3J1, Canada (e-mail: Amine.Maaref@huawei.com).} \vspace{-5mm}}

\maketitle

\begin{abstract}
The era of ubiquitous, affordable wireless connectivity has opened doors to countless practical applications. In this context, ambient backscatter communication (AmBC) stands out, utilizing passive tags to establish connections with readers by harnessing reflected ambient radio frequency (RF) signals. However, conventional data detectors face limitations due to their inadequate knowledge of channel and RF-source parameters. To address this challenge, we propose an innovative approach using a deep neural network (DNN) for channel state estimation (CSI) and signal detection within AmBC systems. Unlike traditional methods that separate CSI estimation and data detection, our approach leverages a DNN to implicitly estimate CSI and simultaneously detect data. The DNN model, trained offline using simulated data derived from channel statistics, excels in online data recovery, ensuring robust performance in practical scenarios. Comprehensive evaluations validate the superiority of our proposed DNN method over traditional detectors, particularly in terms of bit error rate (BER). In high signal-to-noise ratio (SNR) conditions, our method exhibits an impressive approximately 20\% improvement in BER performance compared to the maximum likelihood (ML) approach. These results underscore the effectiveness of our developed approach for AmBC channel estimation and signal detection. In summary, our method outperforms traditional detectors, bolstering the reliability and efficiency of AmBC systems, even in challenging channel conditions.
\end{abstract}


\section{Introduction}

The exponential growth of connected devices in the sixth generation (6G) wireless network poses significant challenges in terms of spectrum and energy usage \cite{Dinh_Nguyen}. One crucial application driving the interconnection of numerous devices is the Internet of Things (IoT) \cite{3gpp-tsg-95e, 3gpp-ran-rel-18}. However, devices powered by batteries with limited storage face a critical obstacle of frequent recharging, which hampers the energy and cost efficiency of 6G \cite{Jiang2018,Fatemeh_Rezaei}. To tackle this challenge, ambient backscatter communication (AmBC) has emerged as a highly promising solution, attracting significant attention in the literature \cite{Nguyen_Van,Rezaei2023,Galappaththige2022,Azar_Hakimi_1}.

In AmBC, a low-power and low-cost device called a tag utilizes ambient radio frequency (RF) signals, such as those from wireless fidelity (Wi-Fi) or television (TV), as its carrier signal instead of generating RF signals on its own. This approach eliminates the need for power-hungry active components like oscillators, mixers, and amplifiers, resulting in power consumption of only a few tens of $\mu$-watts, which is significantly lower compared to active mobile devices consuming 1000 times more power \cite{Jiang2018}. These maintenance-free tags can collect and backscatter sensed data under a wide range of environmental conditions, spanning from extreme scenarios like high-pressure or toxic environments to moderate conditions in farmlands \cite{3gpp-tsg-95e, 3gpp-ran-rel-18}. They find applications in various industries, such as industrial automation or smart logistics, where they can efficiently track millions of parcels \cite{3gpp-tr-22.840}. Another important use case for these tags is in smart agriculture, encompassing applications ranging from soil monitoring to livestock tracking. The tags utilize load modulation to alter the phase, amplitude, or frequency of the RF signal, enabling the modulation of data over the radiated RF source \cite{Nguyen_Van}.

On the other side, the reader (backscatter receiver) must decode the tag signal in order to extract the tag data. To ensure the accuracy and reliability of this task,   the reader must eliminate interference caused by external environmental signals \cite{Donatella_Darsena}. Thus, signal detection presents several challenges.  These challenges include:

\begin{itemize}
    \item\textbf{ Low backscatter signal strength:} The tag-reflected signal experiences deeper fades due to double path losses, resulting in a low received signal power or signal-to-noise ratio (SNR) at the reader. Moreover, the direct-link signal is typically strong and can cause direct interference to the reader \cite{Azar_Hakimi_1}. Additionally, other ambient RF signals can also corrupt the received signal, making detection even more challenging.  Subsequently, various strategies to solve this challenge have been developed \cite{Chong_Zhang,Morteza_Tavana}. 

    \item \textbf{Unknown parameters:}  The ambient RF source parameters, such as its bandwidth, transmit power, and waveforms, are typically not known. Thus, the cancellation of the direct-link interference (DLI) from the RF source to detect the tag signal is highly challenging.

    \item \textbf{Lack of Channel state information (CSI):} Another challenge arises from the ambiguity in acquiring channel state information (CSI) \cite{Yunkai_Hu_2}. Direct link channel estimation (CE) poses difficulties due to the absence of pilot signals from the ambient RF source. Moreover, the limited memory and simplicity of the tag device prevent it from dispatching an adequate number of pilots to facilitate backscatter link CE. Consequently, several studies have focused on addressing both CE and detection for \abc \cite{Donatella_Darsena}.   
\end{itemize}  

\subsection{Backscatter Signal detection methods  }
In the literature, numerous detection strategies have been developed for \abc signals \cite{Jing_Qian_2, Jing_Qian, Tao2020}. They can be broadly classified into coherent, noncoherent, and machine-learning-based.

 \begin{enumerate}
\item Coherent detection: Requires exact carrier phase knowledge and CSI, offering optimal error probability. However, obtaining CSI and carrier phase knowledge can be challenging in practice.

\item Non-coherent detection: Does not require carrier phase and CSI knowledge, reducing receiver complexity but sacrificing spectral efficiency or performance.

\item Semi-coherent detection: Combines both coherent and non-coherent aspects, using a limited number of training symbols to estimate required parameters without full CSI estimation.
\end{enumerate}

Coherent detection improves the sensitivity of the receiver and enables the detection of weak signals in the presence of noise \cite{xu2019sixty}.  For instance, \cite{Jing_Qian_4} derives a maximum likelihood (ML) detector, characterizing the outage probability. Similarly, \cite{Tao2020} derives the maximum a-posteriori (MAP) detector for  OOK tag modulation. A closed-form expression for the BER of this optimal detector is also derived.

Noncoherent detection, on the other hand, recovers data based on the statistical properties of the received signal.  For example, \cite{Jing_Qian} develops an \abc  noncoherent detector using the generalized likelihood ratio test (GLRT). The joint probability density function (PDF) of the incoming signal is examined in \cite{Sudarshan_Guruacharya} to investigate two non-coherent detectors.  To overcome the lack of training symbols \cite{Jing_Qian_2, Wang2016}, consider the differential encoder at the tag.   Furthermore, \cite{Wang2016}  suggests two detection thresholds, one of which provides roughly minimum BER and the other of which produces balanced error probability for detecting the tag bit. In addition, some IoT devices may be deployed in high-mobility situations, which will cause a shorter channel coherence time and a larger Doppler dispersion than in static scenarios. The authors in \cite{Kartheek_Devineni_2} examine the case of non-coherent detection of ambient signals in a time-selective fading channel via a model of a first-order autoregressive process. 

However, although non-coherent detection eliminates the use of  CSI, optimal detection requires precise CSI. For example,  \cite{Jing_Qian_2} designs a joint-ED and derives the detection threshold, which requires the estimation of some parameters. Next, in \cite{Donatella_Darsena},  the joint \abc CE and detection problem is studied where the full-duplex (FD) orthogonal frequency-division multiplexing (OFDM) access point (AP), and the intended recipient of the backscatter information, are incorporated. The authors leverage the cyclic prefix (CP) structure of OFDM symbols from RF sources to eliminate DLI at the reader. In addition, they solve the detection problem by using the space alternating generalized expectation maximization (SAGE) algorithm. As outlined in \cite{Mohamed_ElMossallamy_3}, the authors posit that backscattered signals can be modulated to operate within a different frequency band in an OFDM-AmBC system, thereby avoiding DLI from the RF source. In contrast, \cite{Youyou_Zhang} proposes a more straightforward approach known as the direct-link averaging detector (DL-AD) to eliminate DLI based on the log-likelihood ratio test while employing a semi-blind channel estimator. However, it is worth noting that these DLI cancellation techniques necessitate more intricate circuitry either at the tag or the reader, which may contradict the cost and energy-efficiency goals of passive backscatter systems. Furthermore, perfect DLI cancellation requires precise time and frequency synchronization, which presents its own set of challenges. 

Machine learning has recently become an active signal detection approach due to its ability to identify patterns in large datasets that are not easily detectable through conventional methods. By leveraging these patterns, machine learning algorithms can significantly enhance signal detection accuracy and precision. For instance,  \cite{Qianqian_Zhang_1}  applies unsupervised learning to detect tag data by extracting signal features based on energy information and grouping them into clusters. This method is further improved by transmitting labeled bits from the tag for cluster-bit mapping. Similarly, \cite{Yunkai_Hu_2} transforms binary phase-shift keying (BPSK) tag signal detection into a supervised machine learning classification problem, outperforming the traditional  MMSE detector.

Another example is the use of Hadamard codes to investigate the detection of BPSK tags \cite{Xiyu_Wang}. For multi-antenna readers, the approach is based on k-nearest neighbors (KNN) classification, where the first step involves eliminating direct interference from the RF source. The remaining signals undergo further signal processing, followed by learning detection and decoding in the third and fourth phases. Moreover, \cite{Chang_Liu} proposes the deep transfer learning (DTL) approach in a multi-antenna AmBC setup. They adopt a conventional neural network to extract data features from the formation matrix and use a covariance matrix aware neural network (CMNet) that is DTL-oriented to detect tag signals. Overall, these studies demonstrate the effectiveness of machine learning-based methods in improving the accuracy and precision of signal detection. Reference \cite{Qianqian_Zhang} introduces a label-assisted transmission framework for AmBC in IoT, eliminating the need for CSI estimation. The paper offers two detection methods using labeled signals and a mix of both labeled and unlabeled signals, both of which match the performance of perfect CSI detectors. Meanwhile, \cite{GuoZXL19} proposes cognitive AmBC for spectrum sharing. Addressing challenges posed by direct link interference from legacy systems, this paper suggests detectors leveraging multiple antennas and presents beamforming and likelihood-ratio-based detectors. A statistical clustering framework for CSI learning and backscatter detection is also introduced, with simulations showing these methods outperform traditional energy detectors (ED).

\subsection{Motivation and Contributions} 

Recently, deep neural networks (DNNs) have emerged as promising solutions for wireless CE  and data detection \cite{Braud2021,Schmidhuber}. They offer several advantages. Firstly, they can understand complex relationships between received signals and transmitted data, resulting in higher detection accuracy. Secondly, they reduce the computational complexity of signal detection by exploiting the parallel processing capabilities of DNNs and offline training. Finally, offline training reduces the energy consumption of the reader, making it more energy-efficient.

Thus, using DNNs for joint CE and detection offers the potential to enhance \abc performance. However, this approach has been underexplored in the context of AmBC. Prior work, such as \cite{Donatella_Darsena}, focused on monostatic backscatter systems with an integrated reader and AP, rather than true AmBC. Therefore, this paper investigates  CE and data detection in AmBC, presenting several significant contributions:
\begin{itemize} 
\item  Firstly, we propose an innovative DNN-based approach to estimate channel coefficients and extract tag symbols in AmBC systems. While the DTL approach was suggested in \cite{Chang_Liu}, we introduce a groundbreaking DNN-based method that outperforms existing methods in accurately estimating channel coefficients. It leverages the advanced capabilities of deep learning to adapt to varying channel conditions, making it more robust and efficient than conventional methods. By sufficiently training the DNN, our approach can learn the complex relationships between the received signals and the transmitted data, and accurately estimate channel coefficients, even in noisy and unpredictable environments. Our experiments demonstrate that the proposed DNN-based method significantly outperforms existing methods in terms of accuracy and robustness.

\item Secondly,  our work also makes a significant contribution by addressing CSI  acquisition using DNN methods. This aspect distinguishes our work from previous studies such as \cite{Qianqian_Zhang_1, Yunkai_Hu_2, Xiyu_Wang}. Previous research mainly focused on using machine learning techniques for detecting the tag signal but did not give enough attention to the critical problem of CSI acquisition. Our approach uses a joint optimization framework that simultaneously estimates the CSI and extracts the tag symbol using the proposed DNN architecture. Through this joint approach,  we can enhance the accuracy of data detection in AmBC systems. 
\item Thirdly, we use the fully connected network (FCN) architecture for our DNN approach due to its simplicity, transparency, and the foundational understanding it offers in this emerging field \cite{Hao_Ye,Braud2021,Xiang}. FCNs, despite their basic structure, are robust and have consistently shown strong performance, a fact reinforced by our experimental results. They provide computational efficiency, crucial for real-time AmBC applications while allowing scalability based on data availability. The uniqueness of our approach is further highlighted through its enhanced performance over conventional methods in AmBC.

\item Fourthly, we present an extensive comparison of the proposed detection-estimation method with traditional approaches, such as ML, semi-coherent (SemiCoh), energy, Bayesian, and GLRT detectors. This comparative analysis allows us to showcase the superiority of our proposed DNN-based method in terms of accuracy and efficiency compared to established techniques. This comparative study highlights the potential of our proposed method as a promising alternative to traditional approaches in addressing the challenges of AmBC.
\end{itemize}

In summary, our work introduces an efficient approach for joint CE and data detection in \abc. Our experiments validate its effectiveness in acquiring accurate CSI and extracting tag symbols, even in low signal-to-noise ratio (SNR)  scenarios, highlighting its potential for reliable AmBC applications.

\textit{Notation:} Vectors and matrices are represented by boldface lowercase letters and uppercase letters, respectively. For matrix $\mathbf{A}$,  $\mathbf{A}^H$ and $\mathbf{A}^T$  denote the Hermitian conjugate transpose and the transpose. Euclidean norms of complex vectors and absolute values of complex scalars are represented by $|\cdot|$ and $|\cdot|$, respectively. $f_X(.)$ denotes the probability density function of  $X$. The expectation operator is denoted by $\mathbb{E}[\cdot]$. A circularly symmetric complex Gaussian (CSCG) random vector with mean $\boldsymbol{\mu}$ and covariance matrix $\mathbf{C}$ is represented as $\sim \mathcal{C}\mathcal{M}(\boldsymbol{\mu}, \mathbf{C})$. Additionally, $\mathbb{R}^{M\times N}$ and $\mathbb{C}^{M\times N}$ denote   $M\times N$ dimensional real and  complex matrices.

\begin{figure*}[t]
\centering
	\includegraphics[width=7in]{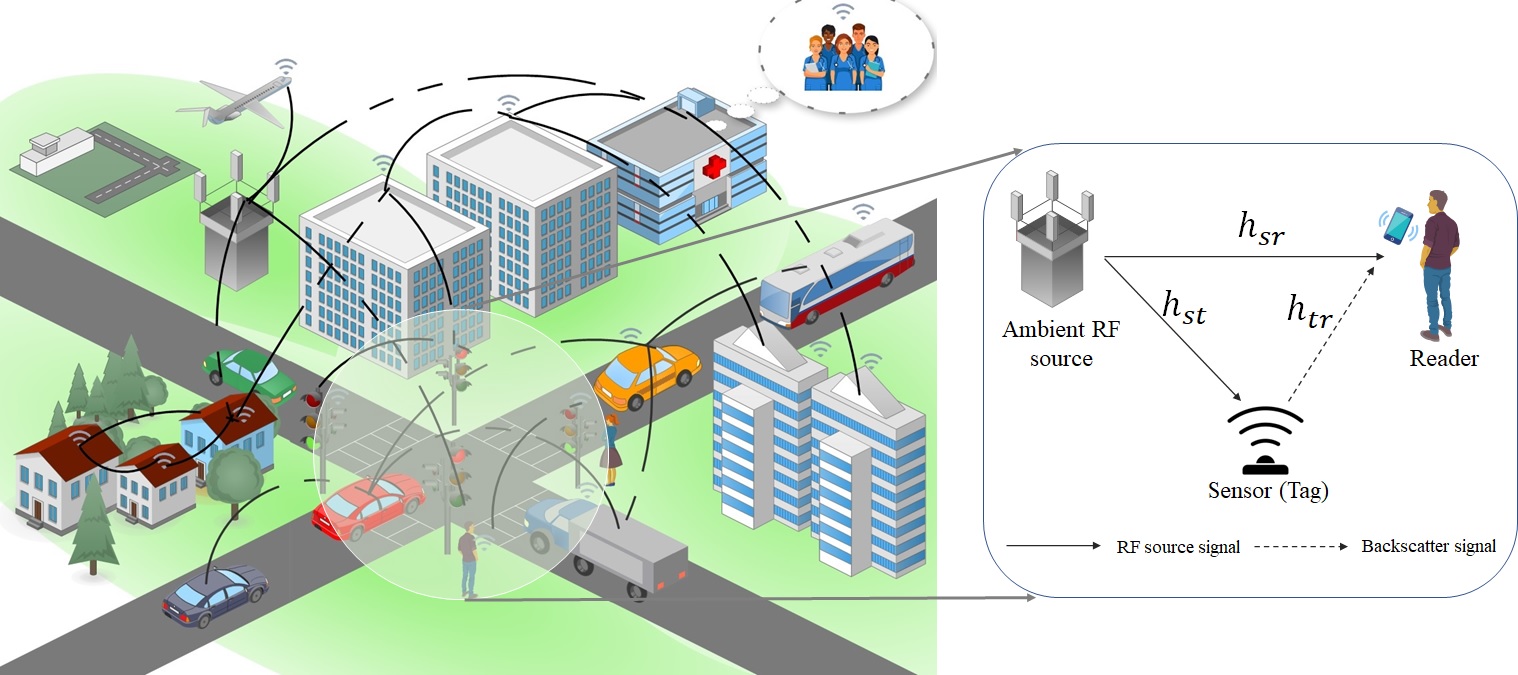}
	\caption{ Smart city IoT network application scenarios.}\label{system_model}
\end{figure*}

\section{AmBC System Model}\label{General_AmBC Systemodel}
The considered network -- Fig \ref{system_model} --  comprises an ambient RF source, a tag, and a reader.  All nodes are equipped with a single antenna \cite{Sudarshan_Guruacharya, Qianqian_Zhang_1}. both the reader and the tag receive the ambient signal from the  RF source.  The tag harvests energy from it and subsequently reflects its data to the reader based on the power splitting ratio, which is defined below.  However, the reader's detection process is affected by the strong direct link interference from the RF source. We next describe the  EH unit and the modulation block of the tag.  

\subsection{Tag's Backscatter Modulation \&  Energy Harvesting} 
\subsubsection{Backscatter Modulation}
A tag is a passive device with no active electronics and cannot generate an RF signal but reflects ambient RF signals opportunistically to send data. To do that, the tag tunes its load impedance depending on its bit sequence, which is called load modulation. For example, the tag realizes OOK by switching between two load impedances to generate bit “$0$'' or “$1$'' by matching or mismatching with antenna impedance, indicating absorbing or reflecting (i.e., non-backscattering and backscattering states), respectively. The  reflection coefficient of the tag is expressed as follows \cite{Nguyen_Van}:
\begin{align}\label{tag:ref}
\Gamma_{i} = \frac{Z_i-Z_{a}^\star}{Z_i+Z_{a}}, 
\end{align}
where $Z_{a}$ denotes the antenna impedance of the tag, which depends on the structure of the antenna, $Z_i$ is the load impedance of state $i = \{1, 2\}$. Here, we assume that the tag uses OOK modulation. It adjusts its impedance to either match or mismatch the incoming signal. More specifically, $ \Gamma_i = |\Gamma_i| e^{j \theta_i}$ where the tag can use distinct  $\theta_i\in \{0, \pi\}$ values to send its data. In addition, the reflection coefficients of impedance values have a constant magnitude, i.e., $|\Gamma_i|^2=|\Gamma|^2=\xi \in (0, 1]$. In particular, $\xi = \vert \Gamma \vert^2$ denotes the power reflection coefficient at the tag satisfying $0 \leq \xi \leq 1$. To design the constellation points, the load impedances $Z_i$ can then be computed via the Smith chart techniques \cite{ReedPH21}. 

\subsubsection{Energy Harvesting}
Each tag performs EH and data transmission simultaneously \cite{Azar_Hakimi}. Thus, the received RF signal power is divided into two parts based on a power splitting ratio, $\xi$. For more details, please see \cite{Rezaei2023,Galappaththige2022} and references therein. While we provide a brief overview of the EH aspect for context, our paper primarily concentrates on signal detection at the reader. The efficiency or intricacies of the EH phase do not have a direct impact on signal detection. However, we recognize the potential implications of EH and plan to investigate its impact in future research. 

\subsection{Channel Modeling}
Backscatter channels demonstrate unique characteristics in contrast to conventional communication channels. Specifically, the backscattered link is the cascade of the two channels from the transmitter to the tag and from the tag to the receiver. As a consequence, it encounters a dual-path loss phenomenon, raising the risk of severe deep fades. These occurrences can lead to communication outages and increased BERs. In our analysis, we assume a scattered environment and model the channel coefficients initially as zero-mean CSCG  random variables. This assumption is akin to adopting the small-scale Rayleigh flat fading model with a predetermined coherence time.
 
This assumption is widely used in AmBc studies for environments without a dominant line-of-sight (LoS) component \cite{Jing_Qian_4,Sudarshan_Guruacharya}. For example, it holds for environments such as factories where numerous obstructions and moving objects cause signal fluctuations. Similarly, during storage or transportation, tags enclosed in packaging or inside storage facilities may encounter multipath propagation due to reflections.

The system consists of three primary channels: the RF source-to-reader  ($h_{sr}$), the tag-to-reader ($h_{tr}$), and the RF source-to-tag  ($h_{st}$). They are characterized by their respective CSCG random variables, given by $h_{sr}\sim\mathcal{CN}(0, \sigma^2_{sr})$, $h_{tr}\sim\mathcal{CN}(0, \sigma^2_{tr})$, and $h_{st}\sim\mathcal{CN}(0, \sigma^2_{st})$. These channels experience independent Rayleigh fading, with their fading coefficients changing independently over distinct coherence time intervals.

We also consider Rician fading when there is the presence of a LoS component along with scattered paths.  Thus, the Rician fading channel model is given by 
\begin{equation}
    h = \sqrt{\frac{\kappa}{\kappa+1}} h^{\text{LoS}}_i + \sqrt{\frac{1}{\kappa+1}} h^{\text{NLoS}}_i,~i\in\{sr, tr, st\},
\end{equation}
where $\kappa$ is the Rician factor,  $h^{\text{LoS}}_i=1$ is the deterministic LoS component that corresponds to the direct path between the transmitter and receiver, without any obstructions or scattering. Also, $h^{\text{NLoS}}_i$ is the non-LoS component that follows the Rayleigh fading model.

\begin{remark}
The chosen tag frame structure is designed for slow-fading channel models, maintaining consistent channel conditions within each frame, a technique prevalent in AmBC studies \cite{finkenzeller2010rfid,Chang_Liu}. Complications emerge with fast-fading channels due to discrepancies in channel environments between training and test data, causing potential misalignment in the neural network's learned features, which could hinder accurate tag signal detection. Addressing this requires redesigning the tag frame structure,  a promising direction for future research. 
\end{remark}

\begin{remark}
Large-scale fading, also known as path loss, plays a crucial role in wireless communication channels for several reasons. These effects affect coverage prediction, interference mitigation, resource allocation, and energy efficiency.  The path loss is superimposed in the small-scale fading.  In this paper, the path loss model is introduced when the received SNRs are determined (see \eqref{pathloss_eq}).  
\end{remark}

\subsection{Signal model}
We represent  the RF source signal  as $s(n)$, satisfying $\mathbb{E}[|s(n)|^2 ]=1$. The signal $s(n)$ could be either a complex Gaussian signal or a modulated one. 
\begin{enumerate}
    \item[(a)] Complex Gaussian ambient source: in wireless communications, complex Gaussian signals, characterized by both magnitude and phase components, often serve as models for interference sources like artificial noise \cite{Lu2022}. These signals can be effectively represented as complex Gaussian random variables, denoted as $s(n)\sim \mathcal{CN}(0, P_s)$, where they exhibit a zero mean and power $P_s$. This assumption enjoys broad relevance in communication systems, as it aligns with the characteristics of numerous modulation schemes, such as OFDM, which manifest near-white Gaussian attributes in the time domain \cite{Shuangqing}. 

    \item[(b)] Modulated ambient source: In contrast, modulated signals are signals that have been altered in a specific way to carry information, e.g., signals of TV towers, cellular BSs, and Wi-Fi APs. Accordingly,  symbol $s(n)$ is assumed to be selected from a  $Q$-ary modulation alphabet, with a power of $P_s$. Thus, it is drawn from a constellation set $\mathcal{S} = \{S_1, S_2, \dots, S_Q\}$, where each symbol is equally likely   \cite{Youyou_Zhang,Qianqian_Zhang}.
\end{enumerate}

Consequently, the  signal received by the tag at the $n$-th sampling instance can be represented as \cite{Sudarshan_Guruacharya}
\begin{equation} 
x(n) = h_{st}s(n),
 \end{equation} 
The tag then reflects  $x(n)$  with reflection coefficient $\Gamma_i$ to communicate its own data. Given that ambient RF sources typically transmit at much higher rates than the tags,  $\Gamma_{i}$ remains constant during the $N$ observations interval \cite{Sudarshan_Guruacharya, Qianqian_Zhang_1}. 
Thus, the backscattered signal from the tag can thus be represented as
\begin{equation}
    x_b(n) = \Gamma_{i} x(n).  
\end{equation}
At the reader, the signal corresponding to the tag symbol is given by
\begin{align}
y(n) &= \sqrt{P_r} h_{sr} s(n) + \sqrt{P_c} h_{tr} x_b(n) + w(n) \nonumber\\
&= \left(\sqrt{P_r}h_{sr}  +  \sqrt{P_c} \Gamma_{i} h_{st} h_{tr}\right)s(n)   + w(n),
\end{align}
where $w(n)$ represents additive white Gaussian noise (AWGN) with mean zero and variance $\sigma^2_{w}$, and the noise samples are assumed to be independent. The average power received directly at the reader is given by $P_r$, while $P_c$ denotes the average power from the backscatter link at the reader.
 
The average received SNRs for both the direct link and the backscatter link can be defined as $\beta_d \triangleq \frac{P_r}{\sigma^2_{w}}$ and $\beta_b \triangleq \frac{P_c}{\sigma^2_{w}}$, respectively. Accordingly, the received power (dBm) at the reader is given by 
\begin{equation}\label{pathloss_eq}
    P_r = P_c + v_{1} \log(d_{st}) + v_{2} \log(d_{tr})-v_{3} \log(d_{sr}) -\log(\xi F G_l^2),
\end{equation}
where $F = \lambda^2/(4\pi)^2$ with $\lambda$ being the wavelength, and $G_l$ is the tag's antenna gain. Also, $d_{sr}$, $d_{st}$, and $d_{tr}$ denote the distances from the RF source to the reader, the RF source to the tag, and the tag to the reader, respectively. The path loss exponents are given by $v_i$, $i\in\{1,2,3\}$. By defining the relative SNR as $\eta \triangleq \frac{\beta_b}{\beta_d} = \frac{P_c}{P_r}$ \cite{Chang_Liu,Qianqian_Zhang}, the received signal at the reader can be represented as
\begin{align}
y(n) &=  \left(h_{sr} + \sqrt{\eta} h_{st} h_{tr}\right)s(n) + w(n).
\end{align}

A lower value of $\eta$ indicates that the backscattered signal is more prominent relative to the direct signal, which can improve the detection performance of the reader. Conversely, a higher value suggests that the direct signal is stronger, leading to increased interference and more challenging detection scenarios. 

\subsection{Signal Detection Problem}
This is a  binary hypothesis testing scenario. The signal received by the reader can be described differently under various hypotheses, as below \cite{kay_detection}:
\begin{equation}\label{received_signal_2}
y(n) =\left\{ \begin{array}{l}
h_{0} s(n) + w(n),~\quad \text{if decide on} ~H_0,\\
h_{1} s(n) + w(n),~\quad \text{if decide on} ~H_1,\\
\end{array} \right.
\end{equation}
Here, $H_0$ corresponds to the null hypothesis and $H_1$ corresponds to the hypothesis that the transmitted symbol is $\Gamma_1$. In addition, $h_{0} = h_{sr}$ denotes the direct channel link and $h_{1} = h_{sr} + \Gamma_{1} h_{st} h_{tr}$ indicates the composite channel link. Based on the hypotheses, the reader can decode the transmitted symbols of the tag. It decides between two hypotheses, i.e., $H_0$ and $H_1$, based on the observed data. To make a decision, two types of information are required \cite{levy2008principles}. Firstly, we need the a-priori probabilities, represented by $\pi_0=P(H = H_0)$ and $\pi_1=P(H = H_1)$, where they should satisfy $\pi_0+\pi_1 = 1$. Secondly, the measurement model for observation vector $\mathbf{y} = [y(0),y(1),\ldots, y(N-1)]^T$, where $N$ denotes the total number of observations, is required. Particularly, the measurement model represents the probability density conditioned on each hypothesis as:
\begin{equation}\label{distribution}
\begin{array}{l}
H_0:  \quad  \mathbf{Y} \sim f_{\mathbf{Y}}(\mathbf{y}|H_0),\\
H_1:  \quad  \mathbf{Y} \sim f_{\mathbf{Y}}(\mathbf{y}|H_1),
\end{array}  
\end{equation}
which are referred to as likelihood functions. To make a decision, the range of $\mathbf{Y}$, denoted as $\mathcal{Y}$, is divided into two decision regions, $\mathcal{Y}_0$ and $\mathcal{Y}_1$, such that if $\mathbf{y} \in \mathcal{Y}_i$, then hypothesis $H_i$ is selected as the best match for the data. Therefore, the design of the decision region is crucial.

\section{Detector Design}\label{detector_design}
 As mentioned before, detection can be classified as coherent, noncoherent, or semi-coherent. However, recent advancements in deep learning have led to the development of DNNs for joint CE and data detection \cite{Hao_Ye,Braud2021}. The DNN can leverage the robust feature extraction and learning capabilities of deep learning models to enhance the accuracy and efficiency of both tasks. Before getting into the DNN, we first describe the most common conventional detectors. 
 
\subsection{Maximum Likelihood Detector}
The ML detector is optimal in a specific sense. Although it offers higher accuracy than an  ED detector, it requires knowledge of the statistical properties of the transmitted signal and noise.  This detector can be developed as follows.  Assuming independence between the sampled signals $y(n)$ and noise $w(n), \forall n$ at the reader, the received signal can be modeled as a Gaussian distribution. The received signal vector $\mathbf{y}$ under hypotheses $\mathcal{H}_0$ and $\mathcal{H}1$ can thus be expressed as
\begin{equation}\label{received_signal_3}
\mathbf{y} \sim\left\{ \begin{array}{l}
\mathcal{CN}(\mathbf{0}, \delta^2_{0} \mathbf{I}_N),~\quad \text{if} \:\:\mathcal{H}_0,\\
\mathcal{CN}(\mathbf{0}, \delta^2_{1}\mathbf{I}_N), ~\quad \text{if} \:\:\mathcal{H}_1,
\end{array} \right.
\end{equation}
where $\delta^2_{0} =|h_0|^2P_s+\sigma^2_w$ and $\delta^2_{1} =|h_1|^2P_s+\sigma^2_w.$ The variances $\delta^2_{0}$ and $\delta^2_{1}$ reflect the combined effects of signal power and noise power for both hypotheses. Consequently, the ML detector involves a likelihood ratio test based on the energy of the received signal vector, denoted by $z=\|\mathbf{y}\|^2$. The ratio is given by
\begin{equation}
    \frac{P(\mathbf{y}|\mathcal{H}_0)}{P(\mathbf{y}|\mathcal{H}_1)}=\left(\frac{\delta^2_1}{\delta^2_0}\right)^N \text{exp} \left(\frac{\delta^2_0-\delta^2_1}{\delta^2_0\delta^2_1}z\right).
\end{equation}
Here, $P(\mathbf{y}|\mathcal{H}_i)$, $\forall i\in \{0,1\}$ represents the  PDF of $\mathbf{y}$ under different hypotheses.  The likelihood ratio only depends on the energy of the received signal vector, allowing for a decision rule based solely on $z$. With equiprobable transmitted messages $\Gamma = \{0, 1\}$, the ML decision rule can be written as follows \cite{Jing_Qian}:
 \begin{equation}\label{ML_detectorr}
   L_{\text{ML}}(\mathbf{y}) = \frac{P(\mathbf{y}|\mathcal{H}_0)}{P(\mathbf{y}|\mathcal{H}_1)} \vc{\gtreqless}{\mathcal{H}_0}{\mathcal{H}_1} 1   \Longrightarrow
  \left\{ \begin{array}{l}
    z \vc{\gtreqless}{\mathcal{H}_0}{\mathcal{H}_1} \Theta_{\text{ML}}^{\text{Th}}, \quad  \delta_0^2 > \delta_1^2, \\
    z \vc{\gtreqless}{\mathcal{H}_1}{\mathcal{H}_0} \Theta_{\text{ML}}^{\text{Th}},  \quad \delta_0^2 < \delta_1^2,
  \end{array} \right.
\end{equation}
where $\Theta_{\text{ML}}^{\text{Th}}$ is the detection threshold, which is given as $\Theta_{\text{ML}}^{\text{Th}} = \frac{N\delta_0^2\delta_0^2}{\delta_1^2-\delta_0^2}\ln \frac{\delta_1^2}{\delta_0^2}$  \cite{Jing_Qian}. 
The decision rule aims to minimize the probability of error by comparing the likelihood of the received signal $ L_{\text{ML}}(\mathbf{y})$ under both hypotheses $\mathcal{H}_0$ and $\mathcal{H}_1$. Specifically, \eqref{ML_detectorr} also indicates that the ML detector can be referred to as a modified energy detection. It is worth mentioning that if $\delta_0^2 = \delta_1^2$, the detection is unsuccessful due to indistinguishable hypotheses. However, the case with $\delta_0^2 = \delta_1^2$ (i.e., $h_{st} = 0$ or $h_{tr} = 0$) is disregarded, as $h_{st} = 0$ indicates a tag backscatter failure and $h_{tr} = 0$ signifies the absence of a received signal from the tag. 

The ML approach provides optimal detection results when the statistical properties of the signals and noise are known. However, perfect CSI may not be available in practical AmBC systems due to the lack of cooperation between the reader and the RF source. This limitation highlights the need for robust detection techniques that can work efficiently in the absence of perfect CSI or under varying channel conditions. Although the knowledge of CSI is unavailable, the values of $\delta_i^2$ can be estimated in a way that will be presented in the next section.

\subsection{Semi-Coherent Detector}
Although the CSI is unknown, it is possible to estimate the values of $\delta^2_i$, $\forall i\in\{0,1\}$ blindly \cite{Jing_Qian}. We term this approach as a SemiCoh detector. Let us briefly describe the steps of blind estimation. The parameters $\delta^2_0$ and $\delta^2_1$ represent the mathematical expectation values of the received signal energy under different hypotheses. We can estimate $\delta^2_i$,$\forall i\in\{0,1\}$  by computing the average energies of a set of received signals with unknown values.

\begin{algorithm}[t]
\caption{Semi-Coherent Detector Algorithm}
\label{alg:channel_energy}
\begin{algorithmic}[1]
\renewcommand{\algorithmicrequire}{\textbf{Input:}}
\renewcommand{\algorithmicensure}{\textbf{Output:}}
\REQUIRE Number of tag symbols $M$, number of training symbols $M_t$, initialize empty list $\Sigma$
\FOR{$m = 1$ \textbf{to} $M$}
\STATE Calculate the normalized energy of $y_m$ as $\Sigma_m = \frac{|\mathbf{y}_m|^2}{N}$ and append to $\Sigma_m to \Sigma$
\ENDFOR
\STATE Sort $\Sigma$ in ascending order to obtain $\Sigma^\uparrow$ and divide into two equal parts: $\Sigma^\uparrow_1$ (first half) and $\Sigma^\uparrow_2$ (second half)
\STATE Compute the averages of the elements in $\Sigma^\uparrow_1$ and $\Sigma^\uparrow_2$ as $\Sigma_\text{min}$ and $\Sigma_\text{max}$, respectively
\STATE Compute the average of $M_t$ normalized powers as $\Sigma_t = \frac{1}{M_t}\sum_{j=1}^{M_t} \frac{|\mathbf{y}_{tj}|^2}{N}$
\IF{$|\Sigma_\text{min} - \Sigma_t| < |\Sigma_\text{max} - \Sigma_t|$}
\STATE Set $\hat{\delta}^2_0 = \Sigma_\text{max}$ and $\hat{\delta}^2_1 = \Sigma_\text{min}$
\ELSE
\STATE Set $\hat{\delta}^2_0 = \Sigma_\text{min}$ and $\hat{\delta}^2_1 = \Sigma_\text{max}$
\ENDIF
\RETURN Estimated values $\{\hat{\delta}^2_0, \hat{\delta}^2_1\}$
\end{algorithmic}
\end{algorithm}

Specifically, let us assume that the channel energy remains constant during $M$ symbol periods of the tag, or correspondingly, $MN$ instances of $s(n)$. The received signal vectors at the reader during this time are denoted as $\mathbf{y}_m$ for $m = \{1, \ldots, M\}$. The SemiCoh signal detection algorithm is presented in \textbf{Algorithm \ref{alg:channel_energy}}. The estimation procedure consists of the following steps \cite{Jing_Qian}:
\begin{itemize}
    \item Calculate the energy of each received signal $\mathbf{y}_m$ and normalize it by dividing it by the number of samples in the signal. This normalization allows for comparison between different received signals. Arrange the normalized energies in ascending order to facilitate their division into two groups, corresponding to the two hypotheses $\mathcal{H}_0$ and $\mathcal{H}_1$.
     \item Since the tag sends ``0" and ``1" with equal probability, we can assume that the first half of the sorted energies corresponds to hypothesis $\mathcal{H}_0$ and the second half to hypothesis $\mathcal{H}_1$. Calculate the average energy for each half, which will provide estimates for $\hat{\delta}^2_0$ and $\hat{\delta}^2_1$. 
     \item Assume the tag sends $M_t \geq 1$ training bits and corresponding received signal vectors are denoted as $\mathbf{y}_{tj}$, $j = \{1, \ldots, M_t\}$. Compute the average normalized power of the $M_t$ training bits, providing an additional signal energy estimate. Combining this information with the results from steps 7 and 8, refine the estimates of $\hat{\delta}^2_0$ and $\hat{\delta}^2_1$.
     \item   These steps provide a practical approach to estimate the required parameters for estimating $\delta^2_0$ and $\delta^2_1$, based on the received signal energies under different hypotheses.   
\end{itemize}

\subsection{Energy Detector}
Ideally, an optimal detector relies on parameters and statistics of the received signal, but these are typically unavailable or require additional work to estimate accurately. Therefore, the ED is proposed in the literature as a method that reduces the need for other parameter values and channel knowledge \cite{Kartheek_Devineni_2,Kang_Lu,5208031,6987540}. The ED test statistic is based on the average energy of the received signal samples \cite{Urkowitz}. The reader determines the transmitted data by averaging the received signal energy over $N$ samples, given as follows:
\begin{equation}\label{E_simple}
\mathcal{E} =\frac{1}{N}\sum^N_{n=1} |y(n)|^2.
\end{equation}
More specifically, an ED is a device used in signal processing to extract the baseband signal from a modulated one \cite{Urkowitz}. This is achieved by removing negative portions of the signal and filtering out the carrier frequency, resulting in a duplicate of the original signal with a direct current (DC) offset. EDs simplify synchronization but are less effective than coherent detection and can be affected by interference and noise. However, they're a viable solution for cost-effective, low-power networks like passive IoT systems. In particular, the probability distribution of $\mathcal{E}$ under different hypotheses can be expressed as
\begin{equation}
\Lambda \sim\left\{ \begin{array}{l}
\mathcal{H}_0:~A_0+B_0,~\quad \text{if} \:\: \Gamma = 0,\\
\mathcal{H}_1:~A_1+B_1, ~\quad \text{if} \:\: \Gamma = 1,
\end{array} \right.
\end{equation}
where the components can be calculated as
\begin{align}
    A_0 &= \frac{1}{N} \sum_{n=1}^{N} |h_0|^2 |s(n)|^2  +|w(n)|^2, \nonumber \\    A_1 &= \frac{1}{N} \sum_{n=1}^{N} |h_1|^2 |s(n)|^2 + |w(n)|^2,\nonumber \\    
    B_0 &= \frac{1}{N} \sum_{n=1}^{N} 2\Re \{h_0 s(n) w^H(n)\}, \nonumber \\     
    B_1 & = \frac{1}{N} \sum_{n=1}^{N} 2\Re \{h_1 s(n) w^H(n)\}.
\end{align}
As $N$ increases, the values of $B_0$ and $B_1$ approach zero as the noise $w(n)$ and the signal $s(n)$ are uncorrelated. According to the central limit theorem (CLT) \cite{Proakis}, $B_0$ and $B_1$ can be denoted as $B_0 \sim \mathcal{N}(0, \gamma_0^2)$ and $B_1 \sim \mathcal{N}(0, \gamma_1^2)$, with their variances given by
$ \gamma_0^2 = \frac{2}{N} |h_0|^2 P_s \sigma^2_{w}$ and $ \gamma_1^2 = \frac{2}{N} |h_1|^2 P_s \sigma^2_{w}$ \cite{Wang2016}. As a result, under different hypotheses, we have
\begin{equation}
\Lambda \sim\left\{ \begin{array}{l}
\mathcal{H}_0:~\Lambda_0   \sim \mathcal{N}({\delta}_0 , \gamma_0^2),~\quad \text{if} \:\: \Gamma = 0,\\
\mathcal{H}_1:~\Lambda_1   \sim \mathcal{N}({\delta}_1 , \gamma_1^2), ~\quad \text{if} \:\: \Gamma = 1.
\end{array} \right.
\end{equation}
Then, the ML decision rule can be expressed as
\begin{equation}
L_\text{ED}(\Lambda) = \Lambda \vc{\gtreqless}{\mathcal{H}1}{\mathcal{H}0} \Theta_{\text{ED}}^{\text{Th}},
\end{equation}
where $\Theta_{\text{ED}}^{\text{Th}}$ is the detection threshold, given by  $\Theta_{\text{ED}}^{\text{Th}} = \frac{{\delta}_0 \gamma_1 + {\delta}_1\gamma_0}{\gamma_0+\gamma_1}$ \cite{Kang_Lu}.
The ED detector is a widely used and straightforward technique for measuring the received signal energy. While it is computationally efficient, it may not provide the best detection accuracy.

\subsection{Bayesian Detector}
A Bayesian detector is a statistical approach based on the Bayes theorem, providing noise resistance and adaptability to environmental changes. Despite its higher computational complexity than ML detector, it excels in incorporating prior knowledge, handling complex signal models, and updating estimates with new observations.  The basic detection rule is as follows. \begin{proposition}
The detector decision can be determined by
\begin{equation}
 L_\text{Baysian}(z) \vc{\gtreqless}{\mathcal{H}1}{\mathcal{H}0} \Theta_{\text{Baysian}}^{\text{Th}},
\end{equation}
where 
\begin{align}
  L_\text{Baysian}(z) &= \log   \int_{\sigma_w^2}^{\infty} \frac{e^{-z/t}}{t^N} \mathcal{I}_1  \left(\frac{t - \sigma_w^2}{P_s}; \sigma_{sr}^2, \xi  \sigma_{st}^2 \sigma_{tr}^2 \right) \mathrm{d} t\nonumber\\
  &- \log \mathcal{I}_N(z; \sigma_w^2, \sigma_{sr}^2 P_s), \nonumber
\end{align}
with an optimal decision threshold $\Theta_{\text{Bayesian}}^{\text{Th}} = \log(\frac{K_0}{K_1})$. 
\end{proposition}
\begin{proof}
    Please see Appendix \ref{app_1}.
\end{proof}
The Bayesian approach typically has higher computational complexity than other detection methods due to the need for integration and calculation of posterior probabilities.

\subsection{GLRT Detector}
The GLRT is another statistical method that involves estimation and detection. The goal of GLRT is to jointly estimate the unknown parameters $(v_0, v_1)$ and replace the unknown parameters with their ML estimates under each hypothesis. The GLRT can provide improved detection accuracy compared to ED and SemiCoh detectors. Also, it is a suboptimal detector that does not require the a-priori probabilities of the unknown parameters. The GLRT technique begins with the calculation of the maximum log-likelihood estimate of the unknown parameter $v$. This estimate, denoted as $v^*$, can be computed using: $ v^* =\underset{v\geq0}{\text{argmax}} ~  \log \text{Pr}\left(\mathbf{y}| v \right)$. After performing some elementary calculus, the ML estimate of $v$ is obtained as follows $v^* = \left(\frac{z}{NP_s} -\frac{\sigma_w^2}{P_s} \right)_+$, where $(x)_+=\max(0,x)$ \cite{Sudarshan_Guruacharya}. Using the estimated value of $v^*$, the system can determine which hypothesis, $\mathcal{H}_0$ or $\mathcal{H}_1$, is more likely to be true. Consequently, the GLRT is defined as follows:
\begin{equation}
L_\text{GLRT}(\mathbf{y}) \overset{\Delta}{=} \frac{\text{Pr}(v^*|\Gamma=1 )}{\text{Pr}(v^*|\Gamma=0)}  \vc{\gtreqless}{\mathcal{H}_1}{\mathcal{H}_0} 1 ,
\end{equation}
where $\text{Pr}(v^*|\Gamma=0 )$ and $\text{Pr}(v^*|\Gamma=1 )$ are given by \eqref{v_0} and \eqref{v_1}, respectively. Subsequently, we can simplify the test statistics as follows:
\begin{equation}
  L_\text{GLRT}(\mathbf{y}) = \frac{v^*}{\sigma_{sr}^2} + \log \mathcal{I}_1 (v^*; \sigma_{sr}^2, \xi \sigma_{st}^2 \sigma_{tr}^2), 
\end{equation}
with an optimal decision threshold $\Theta_{\text{Baysian}}^{\text{Th}}= \log \left(\frac{\xi \sigma_{st}^2 \sigma_{tr}^2}{\sigma_{sr}^2 } \right) - \frac{\sigma_{sr}^2}{\xi \sigma_{st}^2  \sigma_{tr}^2}$. The detector decision can be determined by
\begin{equation}
 L_\text{GLRT}(z) \vc{\gtreqless}{\mathcal{H}1}{\mathcal{H}0} \Theta_{\text{GLRT}}^{\text{Th}}.
\end{equation}
This approach can lead to improved detection accuracy when the channel gains are accurately estimated. Finally, Table \ref{tab:detectors} presents a comparison of different signal detection techniques used in AmBC. A proposed DNN method is also included in the comparison. The table helps to provide an overview of the advantages and disadvantages of different detection methods and can be used to guide the selection of the most appropriate technique for a particular application.

\begin{table*}[t]
\centering
\caption{Comparison of Detectors}
\label{tab:detectors}
\begin{tabular}{|l|c|c|c|c|}
\hline
\rowcolor[HTML]{EFEFEF}
\textbf{Detector} & \textbf{Complexity} & \textbf{Robustness to Noise} & \textbf{Robustness to Environmental Changes} & \textbf{Optimality} \\ \hline
ML & High & Moderate & Moderate & Optimal \\ \hline
Bayesian & High & High & High & Suboptimal \\ \hline
ED & Low & Low & Low & Suboptimal \\ \hline
GLRT & Moderate & Moderate & Moderate & Suboptimal \\ \hline
SemiCoh & Moderate & High & Moderate & Suboptimal \\ \hline
DNN & High & High & High & Suboptimal \\ \hline
\end{tabular}
\end{table*}

\section{ Deep Learning-Based Joint Estimation and Detection}\label{section_IV}
A DNN is a specific type of artificial neural network (ANN) comprising multiple interconnected layers of nodes (also known as neurons) that can learn intricate patterns in data. Deep learning has achieved remarkable accomplishments across various domains, including computer vision, natural language processing, and speech recognition \cite{goodfellow2016deep}. A DNN architecture typically consists of an input layer, one or multiple hidden layers, and an output layer.

\begin{enumerate}
\item \textit{Input layer:} This,  being the first layer of the neural network, accepts input data in various forms such as images, audio, or text. It processes the input data and forwards it to the next layer.

\item \textit{Hidden Layers:} Hidden layers reside between the input and output layers in a neural network, and their quantity and the number of neurons within each layer are determined by the intricacy of the problem. Every neuron in the hidden layer accepts input from the previous layer, calculates a weighted sum of the input, and passes the output through an activation function. This process introduces nonlinearity into the network, allowing it to learn complex associations between input and output variables.

\begin{itemize}
    \item Activation functions are essential components in artificial neural networks, particularly deep learning models, as they introduce nonlinearity into the network. Popular activation functions include the Sigmoid function, which maps input values to the range $(0, 1)$ and is often used in the output layer for binary classification problems, and the Rectified Linear Unit (ReLU) function, which is computationally efficient and helps mitigate the vanishing gradient problem in deep networks \cite{goodfellow2016deep}.

    \item Batch normalization is another technique to enhance DNN training by addressing the internal covariate shift issue that arises when input distributions change due to weight updates in preceding layers, slowing down the learning process \cite{goodfellow2016deep}. It normalizes input features for each layer, achieving a mean of $0$ and a standard deviation of $1$. Therefore, it stabilizes input distributions and accelerates learning.
\end{itemize}

\item \textit{Output:} As the final layer, the output layer generates the output predictions.

\end{enumerate}
\begin{figure}[t]
\centering
	\includegraphics[width=3.5in]{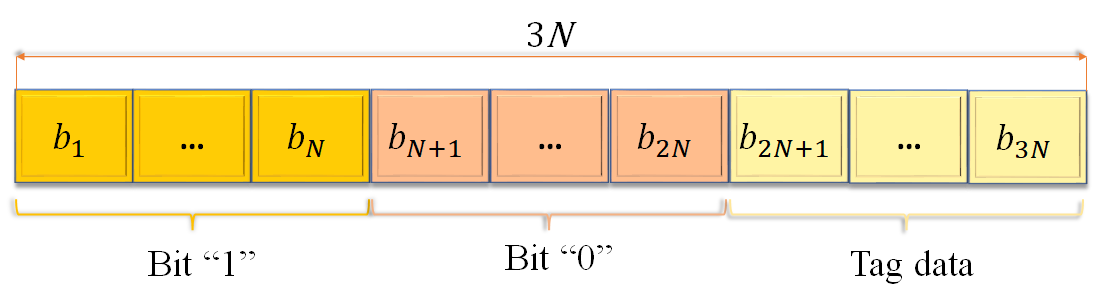}
 	\caption{The frame structure of the tag signals.} \label{frame_structure}
\end{figure}

To obtain an effective DNN model for joint CE and data detection, we first need to design the tag frame structure. The tag frame structure consists of one data symbol and two pilot symbols \cite{Chang_Liu}. The pilot symbols are known to the reader, while one of the remaining tag symbols is employed for data transmission. The signal model frame structure at the reader is depicted in Fig. \ref{frame_structure}, comprising $M_t=2N$ pilot symbols and $3N - M_t$ data symbols in a single frame. Indeed, each tag symbol, whether it is a pilot or a transmitted bit, remains consistent throughout the $N$ RF source symbol periods. 
\begin{remark}
 While backscatter data rates might theoretically match RF sources, they often run into practical challenges in AmBC systems. Notably, high data rates demand swift switching, leading to advanced tag processing, which runs counter to AmBC's core tenets of simplicity and energy efficiency \cite{Kampianakis,Vincent_Liu}. Real-world deployments further validate this approach, as most backscatter devices, such as smart home sensors, inherently prioritize low power over high data rates \cite{ReedPH21,3gpp-tr-22.840}.
\end{remark}
This frame structure is designed for slow-fading channel models where the channel stays constant for each frame. The current random channel is simulated using the Rayleigh fading channel model, and the received signal is acquired by applying channel distortions, including noise, to the tag frames. The training data is gathered by combining the received signal with the original transmitted bit. Subsequently, the real and imaginary components of the tag frames are utilized as input for the DNN. The output layer employs the Sigmoid function to map the results to the $[0, 1]$ range. Then, the DNN is trained on a large dataset of known channel conditions and their corresponding received signals.  The hidden layers perform nonlinear transformations of the input signal to extract features relevant for estimating channel conditions, and the output layer produces an estimate of the transmitted data.

In the subsequent sections, we will delve into a comprehensive examination of the DDN's structure, accompanied by some theoretical analysis.

\subsection{Deep Neural Network Model}
The DNN architecture -- Figure \ref{DNN_architecture} -- consists of three fully connected layers with varying numbers of neurons, specifically $512$, $256$, and $128$ neurons in each layer, respectively. The number of neurons in each layer is selected based on a balance between computational efficiency and model accuracy.

\begin{figure*}[t]
\centering
	\includegraphics[width=5.8in]{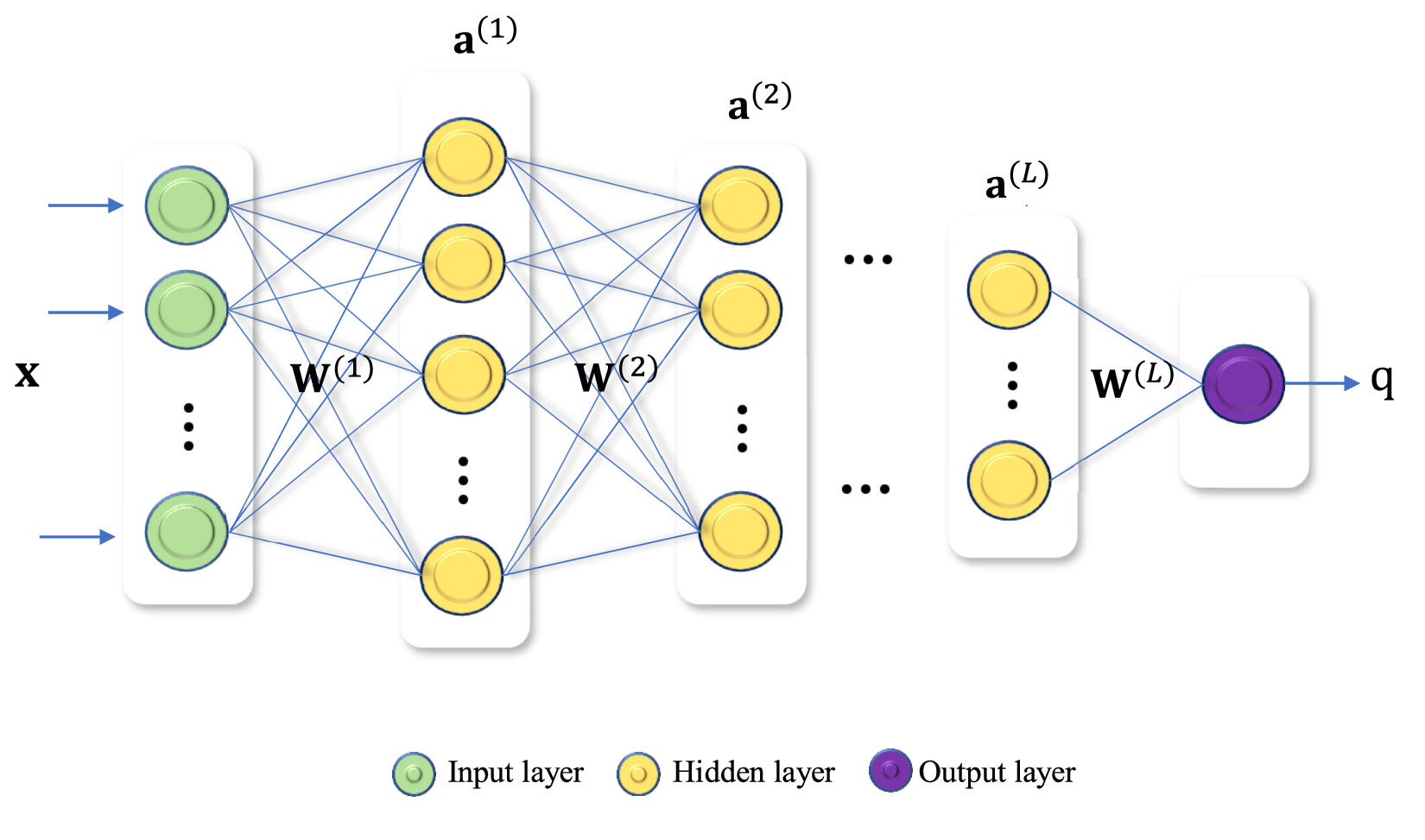}
 	\caption{The deep neural network architecture for AmBC joint channel estimation and signal detection.} \label{DNN_architecture}
\end{figure*}

Let  us  consider the input vector $\boldsymbol{x} \in \mathbb{R}^{4N_T}$ for the DNN, where $N_T = 3N$. This vector includes the real and imaginary parts of the received signal and the original transmitted data, combined with training data. The size of each of these components is $2N_T$. Furthermore, $q \in \mathbb{R}$ denotes the scalar output of the DNN. In this architecture, each layer $l = \{1, \dots, 4\}$ is associated with a weight matrix $\boldsymbol{W}^{(l)} \in \mathbb{R}^{m_l \times m_{l-1}}$ and a bias vector $\boldsymbol{b}^{(l)} \in \mathbb{R}^{m_l}$. The activation function for each layer $l$ is symbolized by $f_l$, while the batch normalization function for each layer $l$ is represented by $g_l$. For each layer $l$, the output can be calculated as
\begin{align}
&\boldsymbol{z}^{(l)} = \boldsymbol{W}^{(l)} g^{(l-1)}(\boldsymbol{a}^{(l-1)}) + \boldsymbol{b}^{(l)}, \quad \boldsymbol{a}^{(l)} = f^{(l)}(\boldsymbol{z}^{(l)}),
\end{align}
where $\boldsymbol{a}^0 = \boldsymbol{x}$ serves as the input vector to the first layer. This computes the activation of each neuron in layer $l$ by taking the weighted sum of the activations from the previous layer, applying the batch normalization function, and passing the result through the activation function.  The final output of the DNN comes from the last layer, denoted as $a^{(L)}$. This represents the probability of each input being classified as either $1$ or $0$, indicating whether the tag is on or off, respectively. In the following, we will introduce the model training framework \cite{Hao_Ye}. In the next section, we will discuss the model training framework in more detail.

\begin{algorithm}[t]
\caption{Deep Learning-Based Signal Detection Algorithm}
\label{alg:DL_signal_detection}
\begin{algorithmic}[1]
\renewcommand{\algorithmicrequire}{\textbf{Input:}}
\renewcommand{\algorithmicensure}{\textbf{Output:}}
\REQUIRE Dataset $D_S = (\boldsymbol{X}_S, Q_S)$ with input-label pairs $(\boldsymbol{x}^{(i)}_S, q^{(i)}_S)$, $i \in \{1, 2, \dots, I_S\}$. Initialize model parameters $\boldsymbol{\phi}$, learning rate $\beta$, batch size $B$, epoch counter $e = 1$, max epochs $E$, and patience $P$ for early stopping. Set $V_{\text{best}}$ to a high value and patience counter $p = 0$.  
 Split dataset into training and validation subsets.
\WHILE{($e \leq E$) and ($p < P$)}
\STATE Shuffle the training set to create random mini-batches of size $B$
\FOR{each mini-batch in the training set}
\STATE Calculate the gradient of the cost function $J(\boldsymbol{\phi})$ in \eqref{J_phi} with respect to $\boldsymbol{\phi}$ using the current mini-batch.
\STATE Update the model parameters $\boldsymbol{\phi}$ by applying the backpropagation algorithm based on the Adam optimizer with learning rate $\beta$.
\ENDFOR
\STATE Evaluate the model performance on the validation set and let $V_{\text{current}}$ be the current validation performance.
\IF{$V_{\text{current}} < V_{\text{best}}$}
\STATE Update $V_{\text{best}} \leftarrow V_{\text{current}}$ and reset the patience counter $p \leftarrow 0$.
\STATE Save the current model parameters as $\boldsymbol{\phi}^*$.
\ELSE
\STATE Increment the patience counter $p \leftarrow p + 1$.
\ENDIF
\STATE Increment the epoch counter $e \leftarrow e + 1$
\STATE  Fine-tune the model by reducing the learning rate $\beta$ if necessary.
\ENDWHILE
\RETURN Model parameters $\boldsymbol{\phi}^*$
\end{algorithmic}
\end{algorithm}

\subsection{Model Training}
The DNN model used for joint Channel Estimation (CE) and data detection follows a two-stage process: (1) offline training and (2) online usage. In the first stage, it is trained using a comprehensive dataset comprising tag symbols. These symbols are collected from a diverse range of information sequences transmitted under various channel conditions, each possessing specific statistical properties. By exposure to this wide array of scenarios, it learns to generalize effectively across different channel conditions.

Once the DNN model is sufficiently trained, it is deployed in the online stage for practical usage. It then takes the received signal as input and generates an output that accurately recovers the transmitted data. The notable advantage of this approach is that it eliminates the explicit need for channel estimation, as the DNN model implicitly learns to perform this task during the training stage. Overall, this approach leverages the power of deep learning to effectively combine channel estimation and data detection, streamlining the process and enhancing overall performance. These two stages are described next. 

Given a dataset $D_S = (\boldsymbol{X}_S, Q_S)$ containing pairs of input data $\boldsymbol{x}^{(i)}_S$ and corresponding labels $q^{(i)}_S$ for $i = \{1, 2, \dots, I_S\}$. We assume that the samples are independent. Let's also assume that the DNN outputs a probability $\hat{q}^{(i)}_S$ for each input $\boldsymbol{x}^{(i)}_S$, representing the likelihood of $\boldsymbol{x}^{(i)}_S$ belonging to class $1$ or $0$. The likelihood function $L( \boldsymbol{\phi})$ denotes the joint probability of observing the labels $q^{(i)}_S$ given the input data $\boldsymbol{x}^{(i)}_S$ and the model parameters $ \boldsymbol{\phi}$ \cite{bishop2006pattern}:
\begin{equation}
\mathcal{L}( \boldsymbol{\phi}) = \prod_{i=1}^{I_S} \text{Pr}(q^{(i)}_S | \boldsymbol{x}^{(i)}_S;  \boldsymbol{\phi}).
\end{equation}
The individual probabilities can be expressed as $\text{Pr}(q^{(i)}_S | \boldsymbol{x}^{(i)}_S;  \boldsymbol{\phi}) = (\hat{q}^{(i)}_S)^{q^{(i)}_S} (1 - \hat{q}^{(i)}_S)^{1 - q^{(i)}_S}$. This equation represents the probability of the true label $q^{(i)}_S$ given the input data $\boldsymbol{x}^{(i)}_S$ and model parameters $\boldsymbol{\phi}$ in a binary classification problem. Here, $\hat{q}^{(i)}_S$ represents the predicted probability of the positive class (class 1) for the $i$-th input sample, and $q^{(i)}_S$ is the true label for the $i$-th input sample, which can be either $0$ or $1$. By substituting this into the likelihood function, we obtain:
\begin{equation}
\mathcal{L}( \boldsymbol{\phi}) = \prod_{k=1}^{I_S} (\hat{q}^{(i)}_S)^{q^{(i)}_S} (1 - \hat{q}^{(i)}_S)^{1 - q^{(i)}_S}.
\end{equation}
It is more practical to work with the log-likelihood function, which is the natural logarithm of the likelihood function:
\begin{equation}\label{l_phi}
\ell(\boldsymbol{\phi}) = \ln \mathcal{L}( \boldsymbol{\phi}) = \sum_{k=1}^{I_S} \left[ q^{(i)}_S \ln \hat{q}^{(i)}_S + (1 - q^{(i)}_S) \ln (1 - \hat{q}^{(i)}_S) \right].
\end{equation}
Indeed, $\ell(\boldsymbol{\phi})$, is used in the context of binary classification problems, as it provides a measure of how well the model predicted probabilities $(\hat{q}^{(i)}_S)$  match the true labels $(q^{(i)}_S)$ of the $i$-th data. To find parameter $\boldsymbol{\phi}$, we need to maximize the conditional PDF given as 
\begin{equation}
  \boldsymbol{\phi}^* = \underset{\boldsymbol{\phi}}{\text{argmax}} ~ \text{Pr}(Q_S | \boldsymbol{X}_S;  \boldsymbol{\phi}),
\end{equation}
which is also equivalent to minimizing the following cost function:
\begin{align}\label{J_phi}
&J(\boldsymbol{\phi}) = -\frac{1}{I_S}  \sum_{k=1}^{I_S} \left[ q^{(i)}_S \ln \hat{q}^{(i)}_S + (1 - q^{(i)}_S) \ln (1 - \hat{q}^{(i)}_S) \right].
\end{align}
During the training process, the DNN model aims to minimize the binary cross-entropy loss by adjusting neuron weights, which ensures that the predicted probabilities are as close as possible to the true class labels. This helps the model learn the optimal parameters $(\boldsymbol{\phi})$ to make accurate predictions. Since the sigmoid activation function is used in the last layer, the output value ranges from $0$ to $1$. To map the output probability to binary class labels ($0$ or $1$), a threshold value is applied to the predicted probability. Typically, the threshold is set at $0.5$. If the output probability is greater than or equal to the threshold (i.e., $\hat{q}^{(i)}_S \geq 0.5$), the input is assigned to class 1. Conversely, if the output probability is less than the threshold (i.e., $\hat{q}^{(i)}_S < 0.5$), the input is assigned to class $0$. This threshold-based approach converts the continuous probability values produced by the sigmoid activation function into discrete binary class labels.

Then, the model parameters are optimized using backpropagation and the Adam optimizer, which adaptively adjusts the learning rate for each parameter based on the first and second moments of the gradients \cite{Diederik}. From the above analysis, the signal detection algorithm utilizing deep learning is outlined in \textbf{Algorithm \ref{alg:DL_signal_detection}}. In addition, the early stopping technique is also used to prevent overfitting, and the learning rate may be adjusted for fine-tuning. The complexity of this deep learning-based signal detection algorithm can be dissected by looking at three crucial factors: the number of epochs ($E$), the number of batches per epoch ($I_S/B$), and the complexity of the backpropagation operation ($\mathcal{O}(N)$), where $N$ represents the total number of parameters in the model. The number of batches per epoch depends on the size of the dataset ($I_S$) and the chosen batch size ($B$). The backpropagation operation, used for updating the model parameters, typically exhibits a complexity proportional to the total number of parameters ($N$) within the model. Consequently, models with a larger number of parameters require a greater amount of computations, thus escalating the time complexity. Given these factors, the algorithm's overall time complexity can be modeled as $\mathcal{O}(E \cdot \frac{I_S}{B} \cdot N)$, indicating that the time complexity increases linearly with the number of epochs, the number of parameters, and inversely with the batch size.

\begin{table}[t]
\renewcommand{\arraystretch}{1.05}
\centering
\caption{Simulation Parameters.}
\label{table-notations}
\begin{tabular}{|l|l|}    
\hline
\textbf{Parameters} & \textbf{Values} \\ \hline \hline
Monte Carlo iterations & \num{1e5} \\ \hline
Speed of light, \(c\) & \qty{3e8}{\m/\s} \\ \hline
Carrier frequency, \(f_c\) & \qty{915}{\mega\hertz} \\ \hline
Tag coefficient, \(\xi\) & 1 \\ \hline
Tag antenna gain, \(G_l\) & \qty{0}{\decibel} \\ \hline
Number of tag bits, \(T_{\text{bit}}\) & 100 \\ \hline
Number of samples, \(N\) & 40 \\ \hline
SNR & \qty{10}{\decibel} \\ \hline
\end{tabular}
\end{table}

\section{Numerical Results}\label{section_V}
In the considered system, three single-antenna nodes are involved, and we set their distances based on prior research \cite{Vincent_Liu,Vikram}. Specifically, the ambient RF source is placed at a distance of $d_{st} = \qty{2.5}{\meter}$ from the tag. The tag itself is positioned at $d_{sr} = \qty{4.8}{\meter}$ from the RF source and at $d_{tr} = \qty{0.5}{\meter}$ from the reader. It is important to note that in Figure 10, we explore various path-loss exponent values, effectively altering the distances. For conciseness, we limit our examination to this specific set of distances and do not report other distance configurations. Furthermore, the reader receives signals from both the direct ambient RF source and the reflected signals from the tag. We compare our  DNN detector against traditional ML, ED, Bayesian,  and GLRT methods.   We also make comparisons to the Gaussian mixture model (GMM), an algorithm using unsupervised learning \cite{Qianqian_Zhang_1}. The GMM leverages the expectation-maximization (EM) algorithm to effectively learn the parameters of the model. It has good robustness and efficiency, particularly for Gaussian data. GMM is used to directly extract energy information from received signals for tag signal detection.  \cite{Qianqian_Zhang_1}. It is assumed that the channel remains constant throughout the $NT_{bit}$ bits \cite{Qianqian_Zhang_1}. Also, the maximum number of iterations for the EM algorithm is set to $1000$.

Following the architectural design (Fig. \ref{DNN_architecture}), the data length is set to $N$, and the pilot number is set to $2N$. The datasets have been partitioned into training, validation, and testing sets, each containing $800,000$, $240,000$, and $240,000$ samples respectively. This allocation ensures a robust base for both model training and performance assessment.  Table \ref{table-notations} summarizes the simulation parameters \cite{Sudarshan_Guruacharya,Vincent_Liu}.

\begin{figure}[t]
\centering
	\includegraphics[width=3.5in]{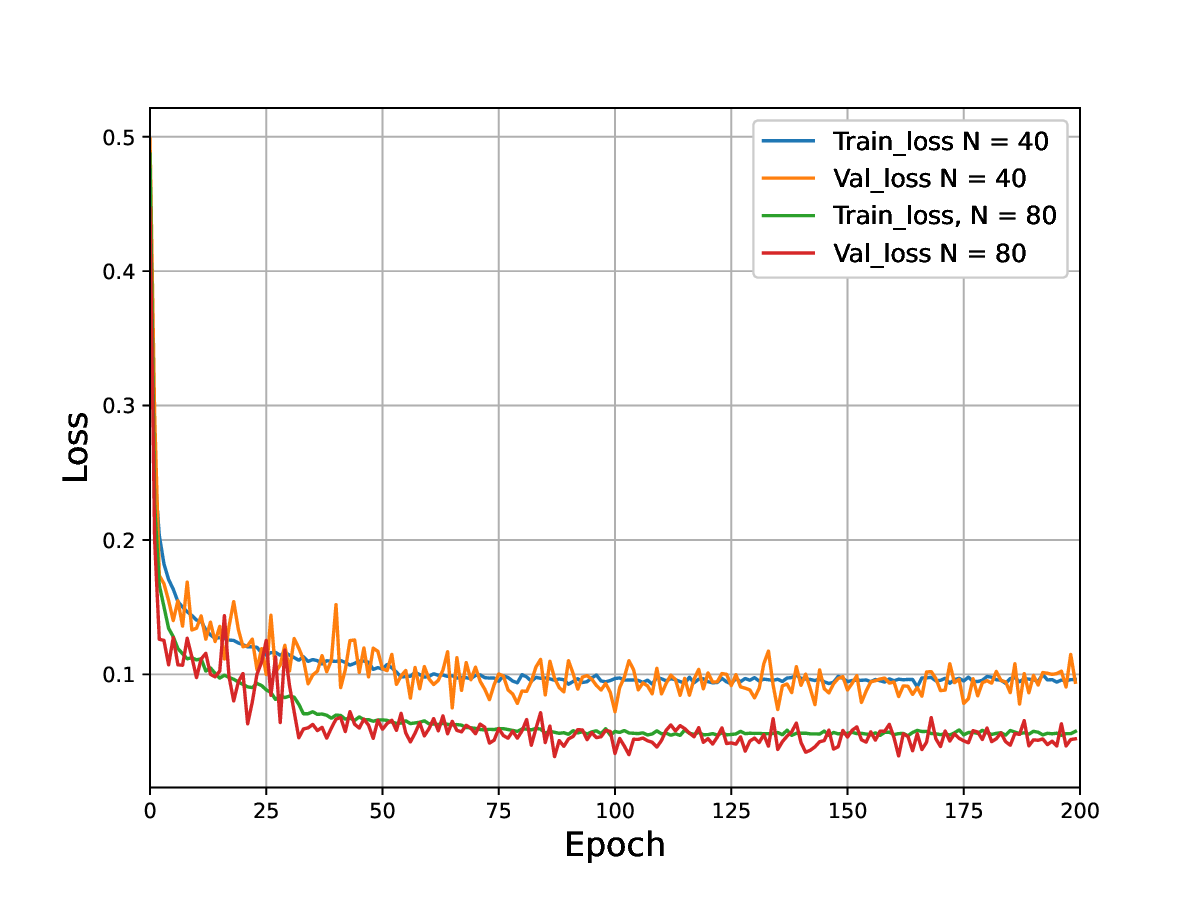}
 	\caption{Training and validation losses.} \label{loss} 
\end{figure}

In Fig. \ref{loss}, we present the training and validation losses as a function of the number of epochs for  $N$ set to $40$ and $80$. The figure illustrates the convergence behavior of the model during the training process over $200$ epochs. As the number of epochs increases, both the training and validation losses decrease, demonstrating the effectiveness of the learning process in optimizing the model parameters. A notable observation from the figure is the impact of increasing $N$ on the loss values. When $N$ increases from $40$ to $80$, both the training and validation losses show a noticeable reduction. This suggests that the model benefits from a larger input size, allowing it to more accurately estimate the channel and detect the data. Consequently, it highlights the selection of an appropriate value for $N$. In both cases, the losses eventually stabilize at values below $0.1$, suggesting that the model achieves satisfactory performance. This figure provides valuable insights into the model learning behavior, highlighting the advantages of using larger $N$ and confirming the convergence of the training and validation losses over the $200$ epochs.

Table \ref{tab:latency_evaluation} presents the learning runtime of our DNN model in relation to data length, both with and without early stopping (ES). The experiment is conducted on a personal computer powered by an Intel\textsuperscript{\textregistered} Xeon\textsuperscript{\textregistered} CPU clocked at 3.5 GHz. The validation loss is monitored for the early stopping mechanism, and the patience level is set at $5$. As data length increases, so does the learning runtime. However, it is evident that employing early stopping can significantly reduce latency. For instance, at a data length of $100$, early stopping reduces the latency from $27.75$ min to $3.33$ min, highlighting its importance in optimizing computational efficiency. 

\begin{table*}[t]
\centering
\caption{Training Time of the DNN Model.}
\begin{tabular}{|c|c|c|}
\hline
\rowcolor[HTML]{EFEFEF}
\textbf{Data Length, $N$} & \textbf{With ES (min)} & \textbf{No ES (min)} \\
\hline
20 & $1.35$ & $16.42$ \\
\hline
40 & $2.08$ & $19.86$ \\
\hline
60 & $2.27$ & $21.84$ \\
\hline
80 & $3.80$ & $25.07$ \\
\hline
100 & $3.33$ & $27.75$ \\
\hline
\end{tabular}
\label{tab:latency_evaluation}
\end{table*}

\begin{figure}[t]
\centering
	\includegraphics[width=3.5in]{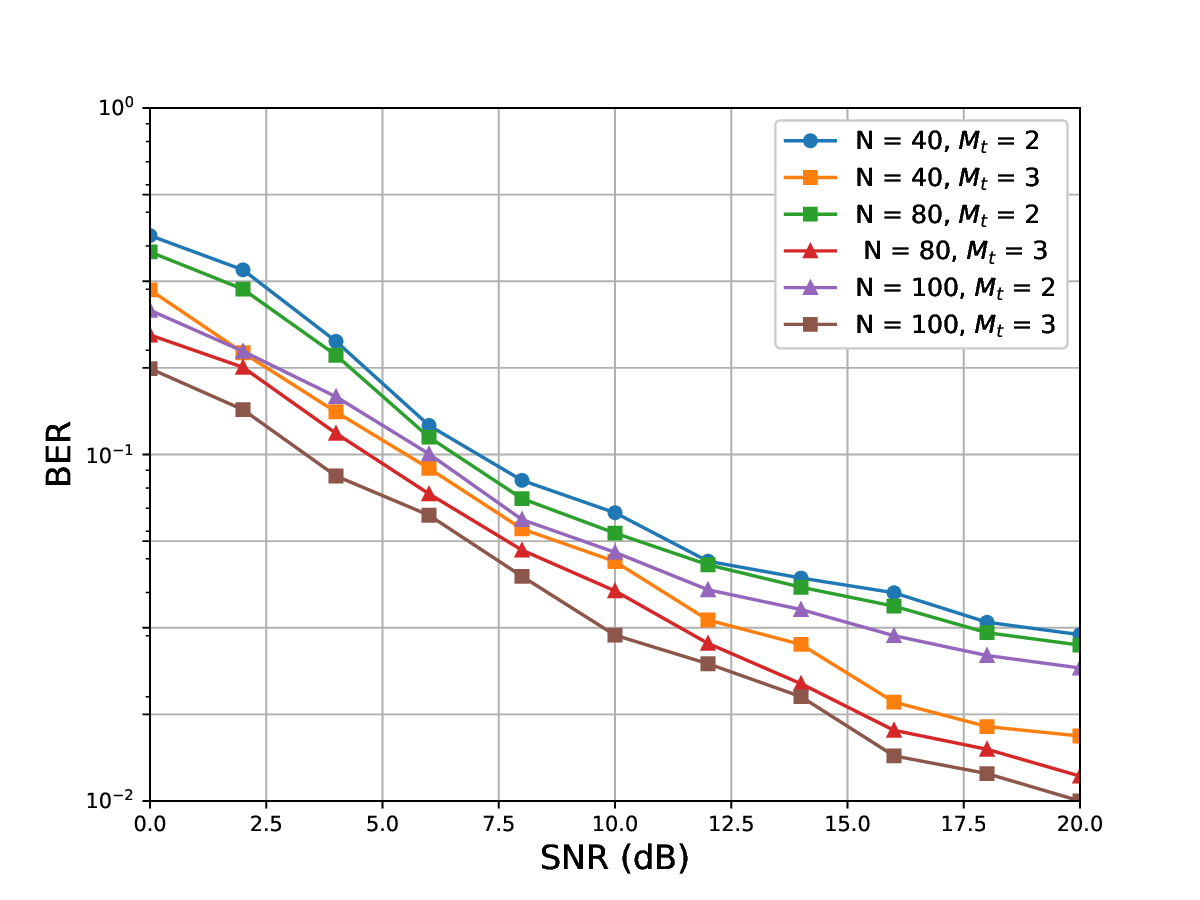}
 	\caption{BER versus SNR for the different number of pilot lengths.} \label{BER_pilot} 
\end{figure}

In Fig. \ref{BER_pilot}, the BER performance of the DNN for joint CE and data detection is plotted against the SNR for various pilot lengths, specifically $M_t = 2$ and $3$. The figure demonstrates a clear trend, where an increase in pilot length leads to better BER performance across the entire SNR range. The pilot symbols serve as a valuable resource for the model to learn the channel characteristics more effectively, and thus, increasing the number of pilot symbols results in more accurate CE. The technical rationale behind this observation can be explained through the increased amount of information available to the DNN for CE when the pilot length is increased. With more pilot symbols, the model is exposed to a larger set of known reference signals transmitted over the channel, allowing it to better understand the underlying channel conditions and adapt accordingly.  Furthermore, as the DNN has a more comprehensive knowledge of the channel conditions, it can more effectively exploit the spatial and temporal correlations present in the wireless communication environment. This exploitation of correlations allows the DNN to make more informed decisions during data detection, thereby reducing the likelihood of errors and ultimately leading to lower BER values. In addition, the enhanced CE achieved through the use of longer pilot sequences also contributes to improved equalization and interference mitigation. As the DNN becomes more adept at discerning the channel nuances, it can more effectively filter out noise and interference, further bolstering the accuracy of data detection.

\begin{figure}[t]
\centering
	\includegraphics[width=3.5in]{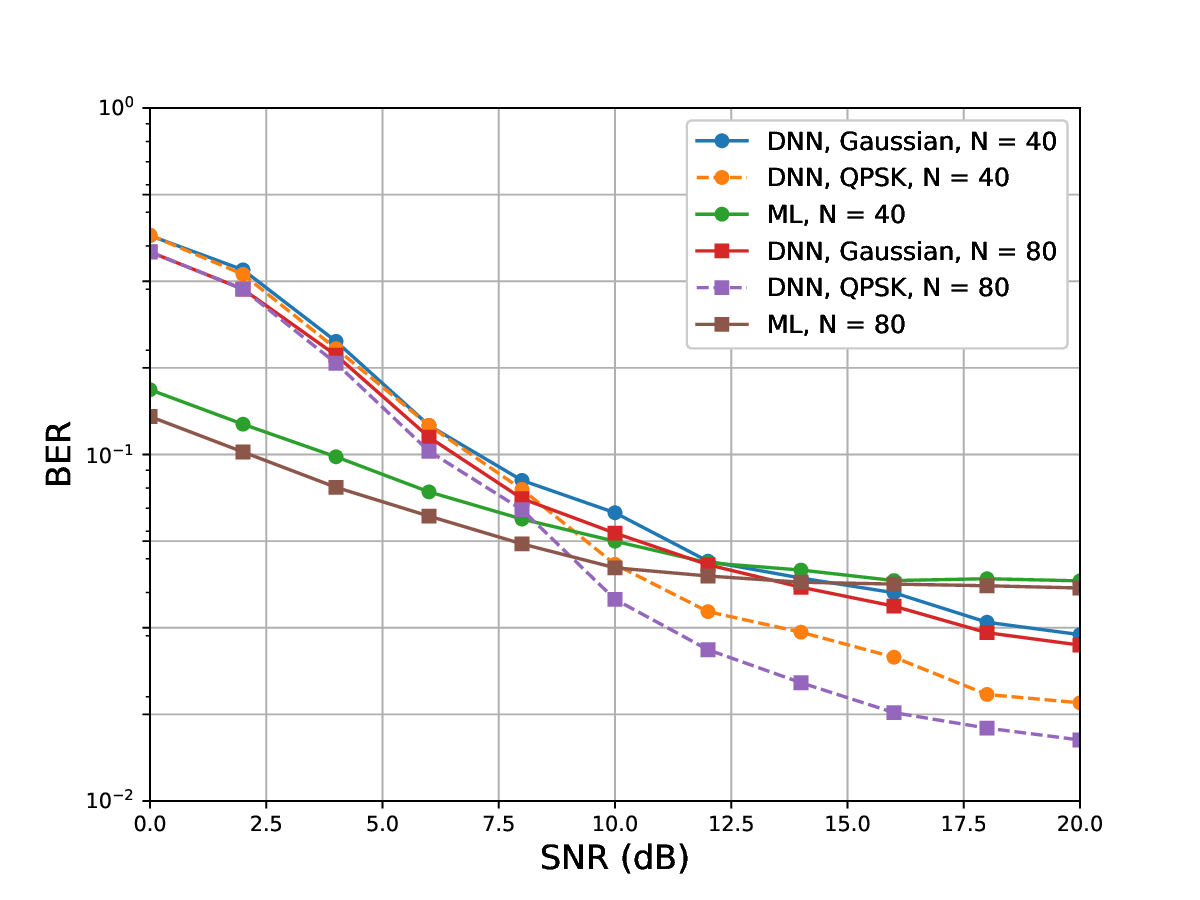}
 	\caption{BER versus SNR for the different number of observation lengths.} \label{BER_N} 
\end{figure}

Fig. \ref{BER_N} showcases the BER performance against SNR for DNN and ML detectors with varying input sizes $N$. This figure demonstrates the effect of increasing $N$ on the BER performance, as well as a comparison between the DNN and ML detectors at various SNR levels. 
The BER results of the introduced algorithms indicate a slightly poorer performance under the complex Gaussian ambient RF source compared to the quadrature phase shift keying (QPSK) ambient one. This is because it is harder to differentiate between the two scenarios in the former situation. It can be observed that the ML detector plateaus at higher SNRs for all $N$ values, indicating a performance limit. This plateau can be attributed to the ML detector's sensitivity to noise and interference, which becomes more noticeable at higher SNRs. However, at lower SNR values, the ML detector exhibits superior performance compared to the DNN. This is because the ML detector likelihood-based approach is more robust to noise at lower SNR levels, leading to better data detection accuracy. The impact of increasing $N$ on the BER performance can also be observed. For both DNN and ML detectors, the BER is improved when the input size $N$ is increased. As $N$ increases, the detectors can better differentiate between the hypotheses, resulting in a lower BER. This improved hypothesis testing can be explained by the fact that a larger input size allows the detectors to observe a greater number of signal samples, which aids in reducing the uncertainty in the decision-making process. This reduced uncertainty leads to a more accurate estimation of the transmitted symbols, ultimately resulting in a lower BER. However, larger input sizes can be used to capture spatial and temporal correlations within the received signal, which can be exploited by DNNs to enhance their estimation and detection capabilities. 

\begin{figure}[t]
\centering
	\includegraphics[width=3.5in]{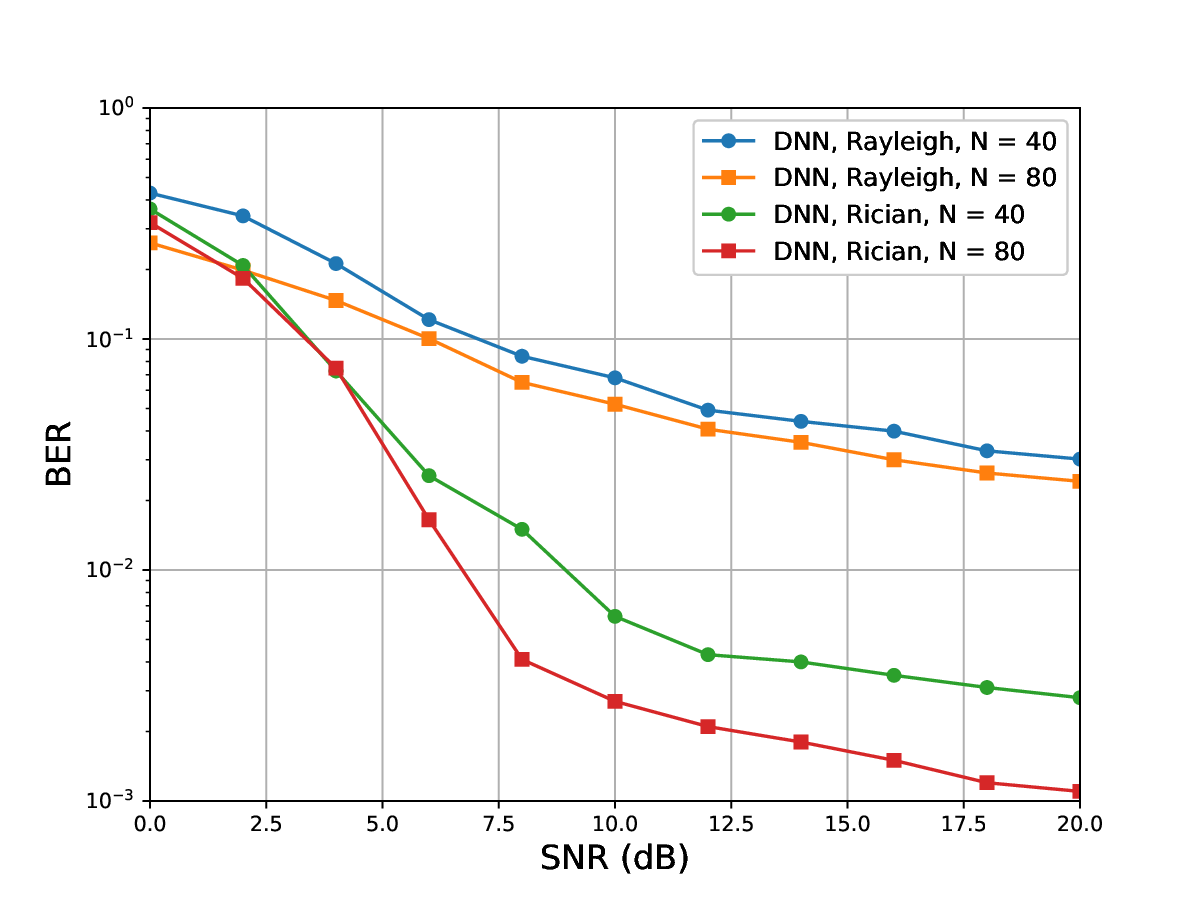}
 	\caption{BER versus SNR for the different channel models.} \label{BER_N_ric} 
\end{figure}

In Fig. \ref{BER_N_ric}, the performance evaluation under varying channel models is depicted, focusing specifically on the Rician fading model with a Rician factor of $\kappa = 3$. The integration of this model becomes indispensable in environments characterized by a pronounced line-of-sight (LoS) path, offering a stark contrast to the Rayleigh model, which predominantly encapsulates scenarios of multi-path propagation absent of a direct path. Due to the existence of the LoS component, the BER is significantly reduced. This reduction in BER is mainly due to reduced vulnerability to interference and fading in the Rician model's direct LoS path. This provides a resilient communication link, capable of mitigating diverse transmission impediments.

\begin{figure}[t]
\centering
	\includegraphics[width=3.5in]{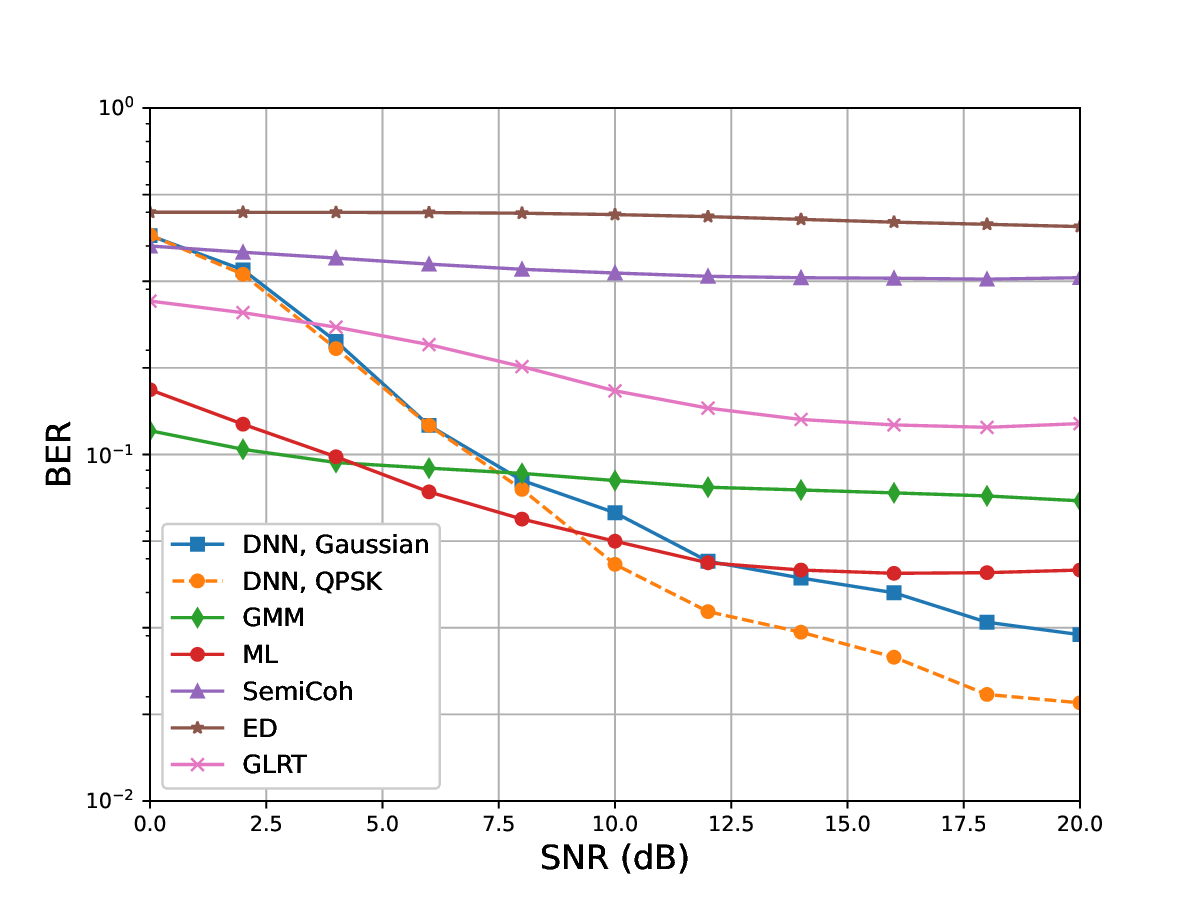}
 	\caption{BER versus SNR} \label{BER_snr} 
\end{figure}

Figure \ref{BER_snr} depicts the BER performance of the detectors across various Signal-to-Noise Ratio (SNR) levels. As expected, the BER decreases as the SNR increases. Comparatively, the SemiCoh detector exhibits a higher BER compared to the ML detector. The ML detector outperforms classical detectors, showcasing superior performance as the SNR increases. However, it reaches a plateau when the SNR becomes relatively large, primarily due to direct link interference. On the other hand, the Bayesian detector, being a more advanced approach compared to the ED, effectively utilizes channel knowledge to enhance detection performance. It employs a maximum log-likelihood approach to estimate channel parameters before conducting the hypothesis test. However, this method requires accurate knowledge of transmit signal statistics and is sensitive to uncertainties, leading to diminished detection performance. Our proposed method outperforms GMM, as the relationships between features are complex and non-linear.

Notably, the DNN demonstrates exceptional performance compared to classical detectors, exhibiting superiority starting from an SNR of $7$ dB and even surpassing the ML detector when the SNR exceeds $12.5$ dB. While ML is considered the optimal approach, it still incorporates interference from the RF source, leading to saturation even with perfect CSI. However, our proposed deep-learning approach excels at extracting complex patterns and improving channel estimation. This finding highlights the effectiveness of the DNN in distinguishing between hypotheses, even in the presence of direct link interference. For instance, at an SNR of $10$ dB, the proposed DNN method achieves a BER of $0.0679$, while the ML detector achieves a BER of $0.0562$, showcasing the advantages of our proposed approach. While it is observed that the DNN has slightly worse BER than the  ML detector, it is important to note that the DNN does not require  CSI. The ML detector, on the other hand, requires perfect channel estimation. Importantly,  DNN executes channel estimation and data detection concurrently.

\begin{figure}[t]
\centering
	\includegraphics[width=3.5in]{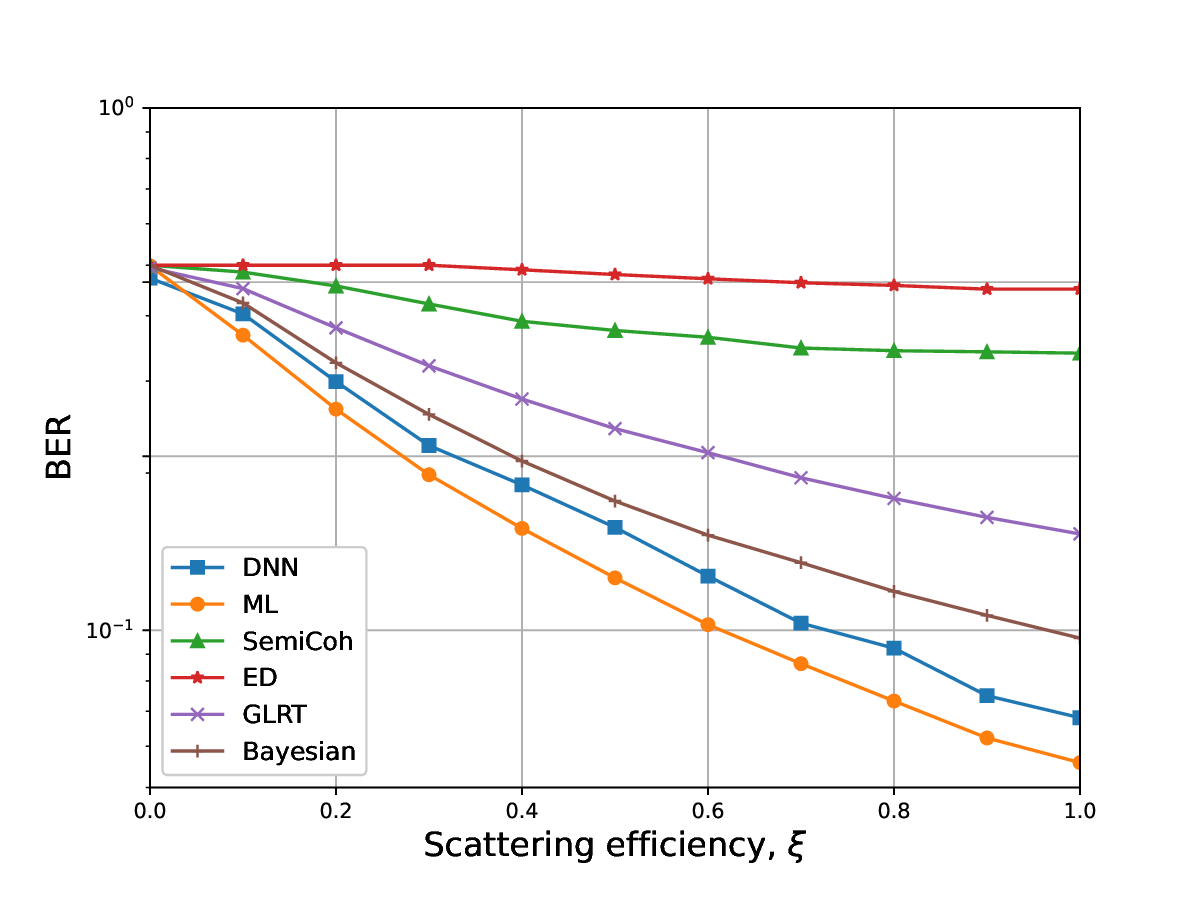}
 	\caption{BER versus scattering efficiency of the tag, $\xi$.} \label{BER_alpha} 
\end{figure}

Fig. \ref{BER_alpha} illustrates the effect of scattering efficiency, $\xi$, on the BER performance of detection schemes, as well as a DNN detector. The figure shows that the BER consistently decreases as $\xi$ increases for all detection schemes under consideration. This can be attributed to the fact that higher scattering efficiency leads to improved signal quality, which in turn enhances the detection capabilities of the detectors.
When $\xi$ is small, the performance of the three proposed detectors (ED, GLRT, and SemiCoh) is relatively similar. This is because, under low scattering efficiency conditions, the received signal quality is poor, making it challenging for the detectors to differentiate between the transmitted symbols. However, as $\xi$ increases, the Bayesian detector exhibits a significantly lower BER than the ED, GLRT, and SemiCoh detectors, indicating its superior performance in scenarios with higher scattering efficiency. This superiority can be attributed to the Bayesian detector's ability to exploit prior knowledge and make more informed decisions during data detection. Furthermore, the performance gap of the ED and SemiCoh detector compared to other detectors also widens, revealing a more pronounced difference in their detection capabilities as $\xi$ increases. This widening gap can be attributed to the inherent limitations of the ED and SemiCoh detectors, which are less capable of handling complex channel conditions and exploiting the benefits of increased scattering efficiency. In addition to these observations, the figure also highlights that the ML detector outperforms all other detectors under consideration, emphasizing its robustness in various scattering conditions. Notably, the DNN detector demonstrates performance close to that of the ML detector. This comparable performance can be attributed to the DNN's ability to learn and adapt to the underlying channel conditions, making it a promising alternative to traditional ML-based detection schemes.

\begin{figure}[t]
\centering
	\includegraphics[width=3.5in]{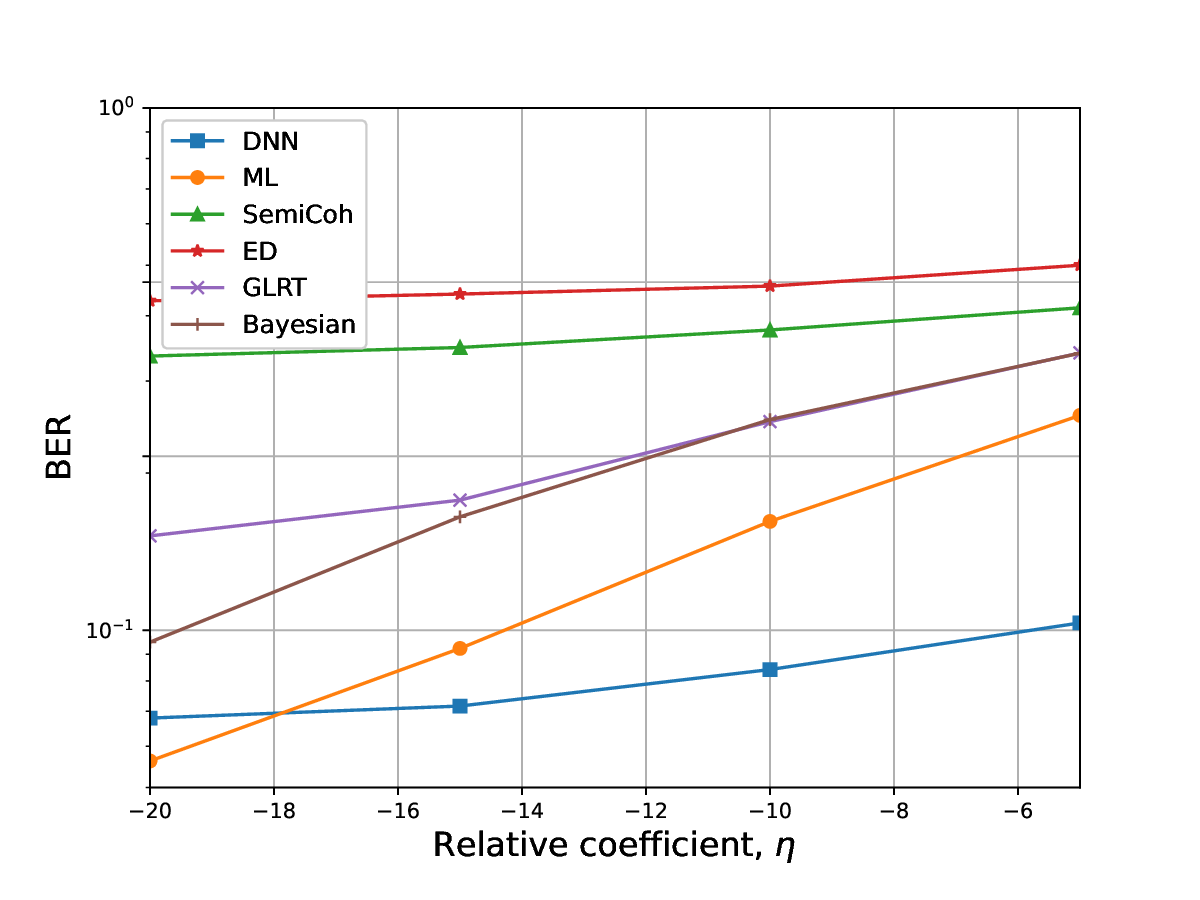}
 	\caption{BER versus relative coefficient between the backscattered signal path and the direct signal path, $\eta$.} \label{BER_zeta} 
\end{figure}

Fig. \ref{BER_zeta} demonstrates the relationship between the BER and the relative coefficient between the backscattered signal path and the direct signal path, $\eta$. The figure reveals that as $\eta$ increases, the BER also increases for all detection schemes. This behavior can be attributed to the fact that as the direct channel link becomes more comparable to the reflected channels, the detection process becomes more challenging. This is due to the growing interference caused by the direct signal path, which can hinder the detectors' ability to accurately estimate the transmitted symbols. Despite this increase in BER, the DNN detector still exhibits superior performance compared to all other detectors, highlighting its robustness and effectiveness in handling challenging detection scenarios. The DNN's strong performance allows it to better cope with the effects of increasing $\eta$. Moreover, the Bayesian and GLRT detectors demonstrate converging performance trends as $\eta$ increases. This convergence suggests that the Bayesian and GLRT detectors may share some similar underlying principles or their performance may be similarly affected by the increasing impact of the direct signal path. One possible explanation for the convergence of the Bayesian and GLRT detectors' performance could be related to their utilization of channel knowledge. As $\eta$ increases, the direct channel link becomes more prominent, and the reflected channels become less distinguishable. Consequently, the additional channel knowledge that the Bayesian and GLRT detectors rely on becomes less informative, resulting in performance degradation.

\section{Conclusion}\label{section_VI}
In this paper, we developed a DNN-based approach for signal detection and channel estimation  \abc systems.  The developed approach benefits from a larger input size and an increase in pilot length leads to better BER performance. Further, we observed superior performance compared to classical detectors, even surpassing the ML detector. The BER performance of the detectors is affected by factors such as scattering efficiency and the relative coefficient between the backscattered signal path and the direct signal path.  The superior performance of the DNN detector highlights its potential for  \abc  data detection, providing a viable alternative to traditional methods. Future works include exploring different neural network architectures, enhancing robustness to varying channel conditions, exploiting multi-antenna systems, and investigating DLI removal.

\appendix

\subsection{Proof of Proposition 1}\label{app_1}
In the AmBC systems, the squared magnitude of channel coefficients can be denoted as $v_0=|h_0|^2$ and $v_1=|h_1|^2$. The PDF of $v_0$ is an  exponential distribution given by
\begin{equation}\label{v_0}
    \text{Pr}(v_0) = \frac{1}{\sigma^2_{sr}} e^{-\frac{v_0}{\sigma^2_{sr}}}.
\end{equation}
On the other hand, the PDF of $v_1$ can be written as
\begin{equation}\label{v_1}
     \text{Pr}(v_1) = \frac{1}{\xi\sigma^2_{st}\sigma^2_{tr}} e^{\frac{\sigma^2_{sr}}{\xi\sigma^2_{st}\sigma^2_{tr}}} \mathcal{I}_1 \left(v_1; \sigma_{sr}^2, \xi \sigma_{st}^2 \sigma_{tr}^2 \right).
\end{equation}
We define the integral function $I_L(z; a, b)$ as follows \cite{Guruacharya_2}:
\begin{equation}
I_L(r; a, b) = \int_0^\infty u^L \exp \left( -\left( \frac{r}{u} + \frac{u}{b} \right) \right) \mathrm{d} u,
\end{equation}
where $r \geq 0$, $a \geq 0$, and $b \geq 0$.  Here, it is assumed that $L$ is a positive integer \cite{Guruacharya_2,Sudarshan_Guruacharya}, although it can be any real number. The conditional PDF of $y(n)$ given $v_0$ and $v_1$ is given by
\begin{align}
    &\text{Pr}\left(\mathbf{y}|v_i \right) = \frac{1}{({\pi (v_i P_s +\sigma^2_w)})^{N}} e^{-\frac{\|\mathbf{y}\|^2}{v_i P_s +\sigma^2_w}}.
\end{align}
The likelihood function of the received signal under $\mathcal{H}_0$ and $\mathcal{H}_1$ can be obtained by integrating out the nuisance parameters $v_0$ and $v_1$, and the posterior probability can be calculated using Bayes' theorem. Given the observation vector $\mathbf{y}$, the Bayesian test is  defined as follows \cite{Sudarshan_Guruacharya}:
\begin{equation}\label{L_bays}
L_\text{Baysian}(\mathbf{y}) \overset{\Delta}{=} \frac{\int_{0}^{\infty} \text{Pr}(\mathbf{y}|v_1)\text{Pr}(v_1) \mathrm{d}v_1}{\int_{0}^{\infty} \text{Pr}(\mathbf{y}| v_0)\text{Pr}(v_0)\mathrm{d}v_0}  \vc{\gtreqless}{\mathcal{H}_1}{\mathcal{H}_0} 1.
\end{equation}
In particular, the numerator and denominator of the Bayesian test can be written as  
\begin{align}
&\int_{0}^{\infty} \text{Pr}(\mathbf{y}| v_0)\text{Pr}(v_0)\mathrm{d}v_0 = K_0 \mathcal{I}_N(z; \sigma_w^2, \sigma_{sr}^2 P_s),\nonumber \\
&\int_{0}^{\infty} \text{Pr}(\mathbf{y}|v_1)\text{Pr}(v_1)\mathrm{d}v_1 = K_1 \int_{\sigma_w^2}^\infty \frac{e^{-z / t} }{t^N} \nonumber \\
&\hspace{40mm} \times \mathcal{I}_1 \left(\frac{t - \sigma_w^2 }{P_s}; \sigma_{sr}^2, \xi \sigma_{st}^2 \sigma_{tr}^2 \right) \mathrm{d} t,
\end{align}
where $K_0 = {\text{exp}\left({\frac{\sigma_w^2}{\sigma_{sr}^2 P_s}}\right) }/{\pi^N\sigma_{sr}^2P_s}$ and $K_1 = {\text{exp}\left({\frac{\sigma_{sr}^2}{\xi\sigma_{st}^2\sigma_{tr}^2}}\right)}/{\pi^N\xi\sigma_{st}^2\sigma_{tr}^2 P_s}$.  
Subsequently, we can simplify the test statistics in \eqref{L_bays} as follows:
\begin{align}
  L_\text{Baysian}(z) &= \log   \int_{\sigma_w^2}^{\infty} \frac{e^{-z/t}}{t^N} \mathcal{I}_1  \left(\frac{t - \sigma_w^2}{P_s}; \sigma_{sr}^2, \xi  \sigma_{st}^2 \sigma_{tr}^2 \right) \mathrm{d} t\nonumber\\
  &- \log \mathcal{I}_N(z; \sigma_w^2, \sigma_{sr}^2 P_s), 
\end{align}
with an optimal decision threshold $\Theta_{\text{Bayesian}}^{\text{Th}} = \log(\frac{K_0}{K_1})$.

\bibliographystyle{ieeetr}
\bibliography{ref}

\end{document}